\documentclass[a4paper,12pt]{article}

\usepackage{amsmath}
\usepackage{amssymb}
\usepackage{enumitem}
\usepackage{amsthm} 
\usepackage{graphicx}
\usepackage{dsfont}
\usepackage[font=small,labelfont=bf]{caption}
\usepackage{fancyhdr}
\usepackage{placeins} 
\usepackage{float}

\usepackage[utf8]{inputenc}
\usepackage{csquotes}
\usepackage[english]{babel}

\newcommand{\bB}{\boldsymbol{B}}
\newcommand{\bF}{\boldsymbol{F}}
\newcommand{\bs}{\boldsymbol{s}}
\newcommand{\bW}{\boldsymbol{W}}
\newcommand{\bx}{\boldsymbol{x}}
\newcommand{\bX}{\boldsymbol{X}}
\newcommand{\by}{\boldsymbol{y}}
\newcommand{\bY}{\boldsymbol{Y}}

\newcommand{\btheta}{\boldsymbol{\theta}}
\newcommand{\bTheta}{\boldsymbol{\Theta}}
\newcommand{\bSigma}{\boldsymbol{\Sigma}}
\newcommand{\blambda}{\boldsymbol{\lambda}}

\newcommand{\bxi}{\boldsymbol{\xi}}
\newcommand{\bzeta}{\boldsymbol{\zeta}}
\newcommand{\bM}{\mathbf{M}} 
\newcommand{\bH}{\mathbf{H}} 

\newcommand{\Xproc}{(\bX_t)_{t \geq 0}}

\newcommand{\Yproc}{(\bY_t)_{t \geq 1}}

\newcommand{\diffd}{\,\text{d}} 

\newcommand{\prodtT}{\prod_{t=1}^T}
\newcommand{\sumtT}{\sum_{t=1}^T}

\DeclareMathOperator{\Expec}{\mathbb{E}}
\DeclareMathOperator{\Plaw}{\mathbb{P}}
\DeclareMathOperator{\Cov}{\text{Cov}}
\DeclareMathOperator{\Var}{\text{Var}}

\newtheorem{lemma}{Lemma}
\newtheorem{theorem}{Theorem}

\usepackage[sort&compress,comma,authoryear]{natbib}

\title{Robust Estimation for Discrete-Time State Space Models}
\author{William H.\ Aeberhard$^{1,2}$, Eva Cantoni$^{3}$, Chris Field$^{2}$,\\ Hans R.\ K\"{u}nsch$^{4}$, Joanna Mills Flemming$^{2}$, and Ximing Xu$^{5}$\\
   \small$^{1}$Department of Mathematical Sciences, Stevens Institute of Technology\\
   \small$^{2}$Department of Mathematics and Statistics, Dalhousie University\\
   \small$^{3}$Research Center for Statistics and GSEM, University of Geneva\\
   \small$^{4}$Seminar f\"{u}r Statistik, ETH Zurich\\
   \small$^{5}$School of Statistics and Data Science, Nankai University}
\date{}

\begin{document}

\maketitle


\begin{abstract} 
State space models (SSMs) are now ubiquitous in many fields and increasingly complicated with observed and unobserved variables often interacting in non-linear fashions. The crucial task of validating model assumptions thus becomes difficult, particularly since some assumptions are formulated about unobserved states and thus cannot be checked with data. Motivated by the complex SSMs used for the assessment of fish stocks, we introduce a robust estimation method for SSMs. We prove the Fisher consistency of our estimator and propose an implementation based on automatic differentiation and the Laplace approximation of integrals which yields fast computations. Simulation studies demonstrate that our robust procedure performs well both with and without deviations from model assumptions. Applying it to the stock assessment model for pollock in the North Sea highlights the ability of our procedure to identify years with atypical observations.
\end{abstract}

\section{Introduction}\label{s:intro}

Beyond the traditional engineering applications following the seminal work of \citet{kalman1960}, state space models (SSMs) are becoming increasingly popular in various fields such as speech recognition \citep{juang1991}, mathematical finance \citep[e.g.][]{geyer1999}, and animal movement modeling \citep{langrock2012}; see \citet{durbin2012} for overviews. This popularity is primarily due to the flexible hierarchical structure of such models, where both observed and unobserved variables are modeled through time.

SSMs are also becoming ubiquitous in fisheries science \citep{aeberhard2018reviewssmfish}, where the task of assessing the state of a fish stock plays a central role in providing science-based advice to policy makers and ultimately ensuring a sustainable management of fisheries \citep[see][]{hilborn1992}. Assessing a fish stock involves predicting some of its time-varying characteristics which cannot be directly measured, such as the age-stratified abundance of fish cohorts and the fishing mortality rate. The state space framework is relevant for distinguishing such unobserved states (akin to random effects or latent variables) from measured quantities, typically in the form of reported commercial catches and indices of relative fish abundance coming from governmental and independent surveys. The assessment of the stock of pollock in the North Sea as conducted by the International Council for the Exploration of the Sea (ICES) exemplifies this well as it relies, among other models, on a SSM \citep{ices2015}. This SSM, often referred to as a SAM (State space Assessment Model) in the fisheries science literature, combines various data sources to attempt to explain and predict the state of the stock. However, whether the intricate and rather strict assumptions of this SSM are satisfied in reality is debatable \citep[e.g.][]{maunder2014}. 
	
As in fisheries science, many modern applications formulate SSMs that are increasingly complicated, with observed and unobserved variables typically interacting in a non-linear fashion. The crucial task of validating model assumptions thus becomes difficult. This concern is all the more serious since the assumed dynamics of unobserved states cannot, by definition, be verified. Evidently, SSMs are a class of models for which robust methods \citep{hampel1986,huber2009} are necessary. Many robust fitting procedures have been proposed for both time series \citep[e.g.][]{kunsch1984,martin1986} and mixed effects models \citep[e.g.][]{welsh1997,moustaki2006}. SSMs can be seen as the intersection of these two classes of models, although the usual devices of robust $M$-estimation are not directly applicable here because of inherent high-dimensional integrals that need to be approximated in general. SSMs have received some attention regarding robust filtering and state prediction based on observations available at some intermediate time given estimates of model parameters that do not change over time \citep{ruckdeschel2010,calvet2015}. Regarding the estimation of model parameters, attention has been confined mainly to the maximum likelihood estimator (MLE). \citet{douc2012} showed that the MLE converges to some value under stationary and ergodic, possibly misspecified, observation processes but that this value only coincides with the true parameter of an assumed model in well-specified cases. To our knowledge, the only proposal for a robust estimator is in \citet{xu2015}, but their approach lacks generality and the estimator is shown to be consistent only as it converges to the MLE. In other words, this estimator cannot be robust and consistent at the same time. The purpose of the present paper is thus to propose a robust and consistent estimation method that is applicable to general SSMs. The main restriction here is that state transition densities must be given analytically. We will focus on the offline estimation of parameters, as the motivating fisheries data are collected yearly, and seek robustness mainly for reliable estimation and inference and for the detection of potentially deviating observations. We will not only address the theoretical properties of our robust estimator, but also propose an effective implementation and discuss computational aspects which are typically challenging for SSMs, even when robustness issues are not considered.

The paper proceeds as follow. Section~\ref{s:notation} presents our notation and defines the MLE for reference. In Section~\ref{s:rob} we detail our proposed robust estimation method, discuss a viable implementation, and study its theoretical properties. Section~\ref{s:sim} presents a simulation study based on the North Sea pollock data. The application of our robust method to the North Sea pollock stock assessment is detailed in Section~\ref{s:data}. Finally we conclude in Section~\ref{s:conc} and discuss open questions and challenges.

\section{Notation and maximum likelihood estimation}\label{s:notation}

A SSM consists of two stochastic processes, $\bX = \Xproc$ and $\bY = \Yproc$. The unobserved state sequence $\bX$ is a (first-order) Markov process taking values in a state space $\mathsf{X} \subset \mathbb{R}^q$. The sequence of observations $\bY$ takes values in the outcome space $\mathsf{Y} \subseteq \mathbb{R}^r$. The $\bY_t$'s are conditionally independent given the state sequence $\Xproc$, with the conditional distribution of $\bY_t$ depending only on $\bX_t$. The SSM is parameterized by the $p$-dimensional $\btheta$, taking values in a parameter space $\bTheta$. We refer to $\btheta$ as the vector of model parameters, in contrast to $\bX$ which is viewed as a sequence of ``random effects'' following a mixed effects terminology. The particular value $\btheta_0 \in \bTheta$ is considered the true value of the data-generating process. The conditional distribution of $\bX_t$ given $\bX_{t-1}$ is assumed time-homogeneous and admits a density function $f: \bTheta \times \mathsf{X} \times \mathsf{X} \to \mathbb{R}_{+}$ with respect to the Lebesgue measure, while the conditional distribution of $\bY_t$ given $\bX_t$ has density $g: \bTheta \times \mathsf{X} \times \mathsf{Y} \to \mathbb{R}_{+}$ with respect to the Lebesgue measure (instead of Lebesgue measures, other dominating measures could be used such as the counting measure in the case of countable spaces). The density of the initial state $\bX_0$ is denoted by $h: \bTheta \times \mathsf{X} \to \mathbb{R}_{+}$. We denote by $\Plaw_{\btheta}$ the marginal distribution of the observation process $\bY$ with parameter $\btheta$, with corresponding expectation $\Expec_{\btheta}$.

Let the subscript ${i:j}$ denote the inclusive range of indices between integers $i$ and $j$, $j \geq i$, so that e.g.\ $\bX_{1:t}=(\bX_1^\top,\bX_2^\top,\ldots,\bX_t^\top)^\top$, where $^\top$ denotes matrix transpose. The joint density of $\bX_{0:T}$ and $\bY_{1:T}$ is
\begin{align}\label{eq:jointlkhd}
p_{\btheta}(\bx_{0:T}, \by_{1:T}) &= h_{\btheta}(\bx_0) \prodtT f_{\btheta}(\bx_{t-1}, \bx_t) g_{\btheta}(\bx_t, \by_t).
\end{align}
Since $\bX$ is not observable, estimation and inference is based on the marginal log-likelihood of the observations
\begin{align}\label{eq:marglkhd}
\ell(\btheta;\by_{1:T}) &= \log \int_{\mathsf{X}^{(T+1)}} p_{\btheta}(\bx_{0:T}, \by_{1:T}) \diffd \bx_{0:T}.
\end{align}
The MLE of $\btheta$ for a given sample $\by_{1:T}$ is then defined as $\hat{\btheta}_T = \arg\max_{\btheta \in \bTheta} \ell(\btheta;\by_{1:T})$. It can alternatively be written as the solution in $\btheta$ to the maximum likelihood (ML) score equations $\partial \ell(\btheta;\by_{1:T})/\partial \btheta = \bs(\btheta;\by_{1:T}) = \boldsymbol{0}$. Interchanging the order of differentiation and integration, it follows that $\Expec_{\btheta}[\bs(\btheta;\bY_{1:T})] = \boldsymbol{0}$, i.e.\ the ML estimating equations are unbiased.

\section{Robust Estimation}\label{s:rob}

\subsection{Definition of robust estimator}\label{ss:robdef}

The typical approach for constructing robust $M$-estimators is by adaptively downweighting observations in the estimating equations corresponding to an initial (consistent) estimator, such as the score equations for the MLE \citep{field1994}. This approach cannot be applied in the state space framework because estimating equations are not available due to the high-dimensional integrals in (\ref{eq:marglkhd}) that need to be approximated in general. Therefore, we propose an approach that directly ``robustifies'' the joint likelihood of $\bX$ and $\bY$.

Related to the $\Psi$-divergence approach of \citet{eguchi2001}, and in some sense to the idea of bounding deviance residuals of \citet{bianco1996}, we introduce a robustified likelihood by applying a specified function $\rho:\mathbb{R} \to \mathbb{R}$, indexed by a tuning constant $c$, on each log-likelihood contribution of $\bX_{0:T}$ and $\bY_{1:T}$:
\begin{align}\label{eq:robjointlkhd}
\tilde{p}_{\btheta}(\bx_{0:T}, \by_{1:T}) &= \exp\left[\rho_c\big\{\log h_{\btheta}(\bx_0)\big\} + \sumtT \rho_c\big\{\log f_{\btheta}(\bx_{t-1}, \bx_t)\big\} + \rho_c\big\{\log g_{\btheta}(\bx_t, \by_t)\big\}\right],
\end{align}
where $c$ may be set to a different value for each of $f_{\btheta}$, $g_{\btheta}$ and $h_{\btheta}$. Throughout, we use the tilde diacritic ($\,\tilde{}\,$) to indicate a robustified function or quantity. For any given $c < \infty$, $\rho_c$ is required to satisfy some conditions which are spelled out in Assumption~\ref{A:rho} in Section~\ref{ss:theory}. Notably, $\rho_c$ is required to be monotonically increasing and have a continuous first derivative bounded within $[0,1]$, see Figure~\ref{fig:rho} for an illustration. The purpose of $\rho$, given a value for $\btheta$, is to diminish the contribution of observations and states whose likelihood is low while leaving the contributions with large likelihood essentially unchanged. The first derivative $\rho'_c(z) = \partial \rho_c(z)/\partial z$ can then be interpreted as a multiplicative weight at the score level, which draws a useful parallel with the usual device of bounding the score function in the construction of robust $M$-estimators for other classes of models \citep[see e.g.][Chapter~3.2]{huber2009}. We present various $\rho$ functions in Section~\ref{ss:rho}. Note that $\tilde{p}_{\btheta}$ addresses not only outlying observations in $\bY$ but also deviations in the dynamics of the states sequence (e.g.\ structural breaks).

By integrating over the unobserved states, the robustified marginal log-likelihood is
\begin{align}\label{eq:robmarglkhd}
\tilde{\ell}(\btheta;\by_{1:T}) &= \log \int_{\mathsf{X}^{(T+1)}} \tilde{p}_{\btheta}(\bx_{0:T}, \by_{1:T}) \diffd \bx_{0:T}.
\end{align}
The sequence of solutions in $\btheta$ to $\partial \tilde{\ell}(\btheta;\by_{1:T}) / \partial \btheta = \tilde{\bs}(\btheta;\by_{1:T}) = \boldsymbol{0}$ does not converge in general to the true $\btheta_0$ as $T \rightarrow \infty$ if the data exactly come from the assumed SSM. We introduce the robust estimator $\tilde{\btheta}_T$ as the solution in $\btheta$ to the corrected robustified score equations $\tilde{\bs}^\star(\btheta;\by_{1:T}) = \tilde{\bs}(\btheta;\by_{1:T}) - b_T(\btheta) = \boldsymbol{0}$. The correction term
\begin{align}\label{eq:defbprime}
b_T(\btheta) &= \Expec_{\btheta} \left[ \tilde{\bs}(\btheta;\by_{1:T}) \right]
\end{align}
guarantees by construction unbiased estimating equations under the assumed model (and thus Fisher consistency, see Theorem~\ref{th:consistency} below). The definition in (\ref{eq:defbprime}) is the typical representation found in estimating equations for general robust $M$-estimators, and one can check that $b_T(\btheta)$ goes to zero as the tuning constant $c$ goes to infinity. By taking $\lim_{c \rightarrow \infty}$ in (\ref{eq:robmarglkhd}), $\tilde{\ell}(\btheta;\by_{1:T})$ converges to $\ell(\btheta;\by_{1:T})$ and $\tilde{\bs}^\star(\btheta;\by_{1:T})$ converges to the ML score, the latter which has expectation zero under $\Plaw_{\btheta}$ for all $\btheta \in \bTheta$.

\subsection{Choice and tuning of $\rho$ function}\label{ss:rho}

\citet{eguchi2001} proposed a log-logistic $\rho$ function defined as $\rho_c(z) =  \log\left\{ \frac{1+\exp(z+c)}{1+\exp(c)}\right\}$. As $\lim_{z \rightarrow -\infty} \rho_c(z) = -\log((1+\exp(c)))$, for Assumption~\ref{A:rho}\ref{rhocond3} to hold (see Section~\ref{ss:theory}) both $\mathsf{X}$ and $\mathsf{Y}$ should be assumed compact. As an alternative, we introduce the following $\rho$ function which only requires $\mathsf{X}$ to be compact:
\begin{align*}
\rho_c(z) &=  \left\{\begin{array}{ll} c \sinh^{-1}\big((z+c)/c\big) -c & \text{if}\; z < -c\\ z & \text{if}\; z \geq -c, \end{array}\right.
\end{align*}
where $\sinh^{-1}(u) = \log\left(u + \sqrt{1+u^2}\right)$ is the inverse hyperbolic sine function. The first derivative is $\rho'_c(z) = \Big(1+\big((z+c)/c\big)^2\Big)^{-1/2}$ if $z< -c$ and $\rho'_c(z) = 1$ otherwise.
We will refer to this $\rho$ function as the smooth semi-Huber (SSH) function, since its derivative is based on the left-hand side of a smooth version of the Huber weight $\psi_c(u)/u$, where $\psi_c(u) = \max(-c,\min(c,u))$ is the Huber $\psi$ function \citep{huber1964}. The smooth approximation $(1+(u/c)^2)^{-1/2} \approx \psi_c(u)/u$ due to \citet{charbonnier1994} is preferred here to the original Huber weight as the latter is not differentiable everywhere. The effect of the tuning constant $c$ on the SSH $\rho$ function and its derivative is illustrated in Figure~\ref{fig:rho}. On the one hand, a smaller value of $c$ diminishes the contribution of the log-likelihood earlier, yielding more robustness in terms of a smaller maximum asymptotic bias in the presence of contamination. On the other hand, $c \rightarrow \infty$ makes $\rho$ converge to the identity function and the weight represented by $\rho'$ converges to one over the entire real line.

\begin{figure}[ht!]
\begin{center}
\includegraphics[width=0.9\columnwidth]{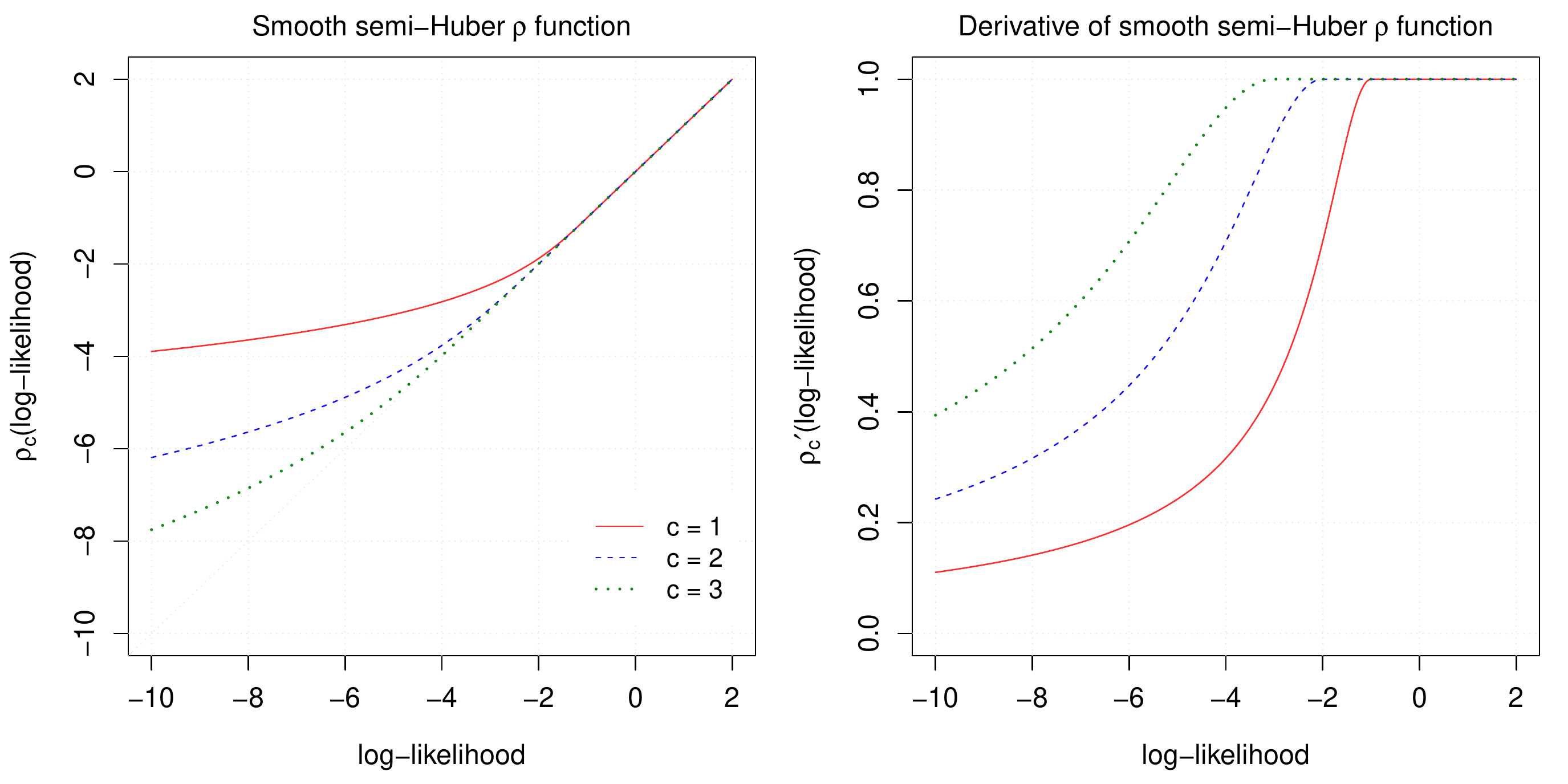}
\caption{SSH $\rho$ function (left panel) and first derivative (right panel) for different values of the tuning constant $c$.
{\label{fig:rho}}%
}
\end{center}
\end{figure}

Among other methods \citep[e.g.][]{rieder2008}, selecting a value for $c$ is typically done so as to achieve a certain efficiency at the assumed model with respect to the MLE, where the efficiency is usually based on asymptotic variances. Rather than relying on asymptotics, we suggest using Monte Carlo simulation instead. The procedure can be summarized as follows: i) For a given value of $c$, compute the robust estimator on the original data, denoted as $\tilde{\btheta}^{\text{original}}$; ii) Use this robust estimate and the design of the original data to simulate $B$ independent samples from the assumed model; iii) On each simulated sample, compute the MLE $\hat{\btheta}$ and the robust estimator $\tilde{\btheta}$; iv) For both estimators, compute the element-wise median absolute deviation (MAD) about $\tilde{\btheta}^{\text{original}}$ over the $B$ replicates, where the MAD for the $j$th element is defined as $\text{median}_{b=1,\ldots,B}\{|\btheta^b_j-\tilde{\btheta}^{\text{original}}_j|\}$ with $\btheta^b$ representing here either estimator on the $b$th sample; v) These two MADs are $p$-vectors and dividing the ML-based one by the robust one element-wise yields empirical efficiencies corresponding to $c$. These five steps are then repeated, by trial-and-error, to find the appropriate value of $c$ for some predefined target efficiency, say 90\%. We advocate using a MAD instead of a more traditional empirical mean squared error (MSE) as the former is a highly robust measure of dispersion and is thus more stable if some estimates among the $B$ replicates (both ML and robust ones) do not converge properly. Also, computing element-wise empirical efficiencies allows to choose different $c$ values among the various likelihood contributions of $f_{\btheta}$, $g_{\btheta}$ and $h_{\btheta}$. This is necessary since the impact of $c$ is directly linked to the magnitude of the likelihood function which depends on scale parameters in $\btheta$.

\subsection{Proposed implementation and computational considerations}\label{ss:implement}

We propose here an effective implementation of the robust estimator. Solving the corrected robustified score equations $\tilde{\bs}^\star(\btheta;\by_{1:T}) = \boldsymbol{0}$ raises two main computational challenges: we need to evaluate the robustified marginal likelihood in (\ref{eq:robmarglkhd}) and the expectation defining $b_T(\btheta)$ in (\ref{eq:defbprime}). Various methods for approximating integrals exist. The choice among simulation-based procedures \citep[e.g.][]{doucet2001book,fearnhead2018} and direct approximations of the integrand \citep[e.g.][]{davis2007} relies on trade-offs between computing time, accuracy, generality, and ease of programming.

Regarding the first integral, we note that even in the case of the MLE, the marginal likelihood in (\ref{eq:marglkhd}) only admits a closed form in very specific cases, namely for a finite state space $\mathsf{X}$ (often referred to as hidden Markov models) and for the linear-Gaussian SSM where all marginals are Gaussian themselves. With the addition of the $\rho$ function, an approximation is indeed necessary and we choose to use Laplace's method \citep[e.g.][Chapter~6.2]{barndorff1989}. It can be summarized as a second-order Taylor expansion of $\log \tilde{p}_{\btheta}(\bx_{0:T}, \by_{1:T})$ about a global maximum for $\bx_{0:T}$, yielding the density kernel of a multivariate Gaussian distribution which can then be integrated following a normalization involving the Hessian determinant. The uncorrected robustified marginal log-likelihood is thus approximated as
\begin{align*}
\tilde{\ell}(\btheta;\by_{1:T}) &\approxeq \tilde{\ell}_{\text{LA}}(\btheta;\by_{1:T}) = \log\tilde{p}_{\btheta}(\hat{\bx}_{0:T}, \by_{1:T}) + \frac{q(T+1)}{2}\log(2\pi) - \frac{1}{2}\log\left\vert - \bH(\hat{\bx}_{0:T},\btheta) \right\vert,
\end{align*}
where the ``LA'' subscript identifies Laplace-approximated quantities, $|\cdot|$ denotes here the matrix determinant,
\begin{align*}
\hat{\bx}_{0:T} &= \hat{\bx}_{0:T}(\by_{1:T},\btheta) = \arg\max_{\bx_{0:T} \in \mathsf{X}^{(T+1)}} \log\tilde{p}_{\btheta}(\bx_{0:T}, \by_{1:T}),
\end{align*}
and $\bH(\bx_{0:T},\btheta) = \partial^2 \log\tilde{p}_{\btheta}(\bx_{0:T}, \by_{1:T})/\partial\bx_{0:T} \partial\bx_{0:T}^\top$. For such an approximation to be valid, this maximum over $\bx_{0:T}$ must exist and be global, or at least dominate others if not unique. The Laplace approximation involves a low computational cost which is crucial given that $\tilde{\ell}(\btheta;\by_{1:T})$ needs to be evaluated many times when optimizing over $\btheta$. To tackle the dependence of $\hat{\bx}_{0:T} = \hat{\bx}_{0:T}(\by_{1:T},\btheta)$ upon $\btheta$ in such optimizations, we use automatic differentiation \citep[AD;][]{griewank2008,skaug2006}. AD allows the numerical evaluation of derivatives without the need for symbolic expressions by propagating derivatives of simple operations through the chain rule. This is not a numerical approximation, the evaluation is exact up to machine precision yet does not require analytical expressions for the gradient. This can be considered state-of-the-art as it is the general device behind the back-propagation algorithm used for gradient descent in the training of artificial neural networks \citep[see e.g.][Chapter~6.5]{goodfellow2016}. The combination of the Laplace approximation and AD is at the core of the R package Template Model Builder \citep[TMB;][]{kristensen2016}, on which our implementation relies. TMB is highly flexible as the user is only required to code the negative log-likelihood $-\log p_{\btheta}(\bx_{0:T}, \by_{1:T})$ or its robustified version. As a useful byproduct, the Laplace approximation provides unobserved state predictions: $\hat{\bx}_{0:T}$ is equivalent to the (posterior) mode of the robustified conditional distribution of $\bX_{0:T}$ given $\bY_{1:T}$ and a value for $\btheta$. The Laplace approximation is increasingly used for evaluating marginal likelihoods in the SSM context, see e.g.\ the recent successful applications in \citet{auger2016}, \citet{thygesen2017}, and \citet{yin2019scallops}. Although its accuracy is yet to be comprehensively studied in this context, we note that simulation-based methods for approximating integrals with arbitrarily good accuracy suffer from other problems such as the degeneracy of particle filters \citep{kantas2015}.

Regarding the second integral, we approximate the expectation defining $b_T(\btheta)$ by Monte Carlo simulation. This is feasible thanks to the efficient implementation in TMB, as both $\tilde{\ell}_{\text{LA}}(\btheta;\by_{1:T})$ and $\tilde{\bs}_{\text{LA}}(\btheta;\by_{1:T}) = \partial\tilde{\ell}_{\text{LA}}(\btheta;\by_{1:T})/\partial\btheta$ can be evaluated instantly. Thus, simulating data at the assumed SSM and evaluating the gradient is computationally inexpensive and requires little additional programming. For a given value of $\btheta$ we generate $B$ Monte Carlo replicates, each denoted $(\bx_{0:T}^b,\by_{1:T}^b)$ for $b=1,\ldots,B$, and obtain
\begin{align}\label{eq:MCapproxFcct}
b_T(\btheta) &= \Expec_{\btheta} \left[ \tilde{\bs}(\btheta;\by_{1:T}) \right] \approxeq \frac{1}{B}\sum_{b=1}^B \tilde{\bs}_{\text{LA}}(\btheta;\by_{1:T}^b).
\end{align}
Our experience is that $B$ needs to be quite large to ensure the approximated $b_T(\btheta)$ does not introduce more bias than there was prior to correction, in particular when the robustness tuning constant $c$ is selected to achieve high empirical efficiency. Based on preliminary simulations with various designs (not presented here), $B=1000$ is probably a minimum while $B=5000$ would be more on the safe side.

We can now present our proposed algorithm for the computation of $\tilde{\btheta}_T$.
\begin{enumerate}[leftmargin=2.4\parindent] 
	\item[Step~1:]{
	Compute the uncorrected robust estimator $\tilde{\btheta}_T^{[1]}$ by maximizing the (uncorrected) Laplace-approximated robustified marginal log-likelihood:
	$$\tilde{\btheta}_T^{[1]} = \arg\max_{\btheta \in \bTheta} \tilde{\ell}_{\text{LA}}(\btheta;\by_{1:T}).$$
	}
	\item[Step~2:]{
	Evaluate the corrected score of this initial estimate:
	$$\tilde{\bs}^\star(\tilde{\btheta}_T^{[1]};\by_{1:T}) = \tilde{\bs}(\tilde{\btheta}_T^{[1]};\by_{1:T}) - b_T(\tilde{\btheta}_T^{[1]}),$$ where $b_T(\tilde{\btheta}_T^{[1]})$ is approximated by Monte Carlo following (\ref{eq:MCapproxFcct}).
	}
	\item[Step~3:]{
	Correct the estimator with a single Newton-Raphson (NR) step:
	$$\tilde{\btheta}_T = \tilde{\btheta}_T^{[1]} - \left.\left(\frac{\partial \tilde{\bs}(\btheta;\by_{1:T})}{\partial \btheta}\right\vert_{\btheta=\tilde{\btheta}_T^{[1]}} \right)^{-1}\tilde{\bs}^\star(\tilde{\btheta}_T^{[1]};\by_{1:T}).$$
	}
\end{enumerate}
A few aspects of this algorithm warrant some discussion. First, as a unique maximum is not guaranteed in general in Step~1, we suggest starting from many (random) values and retaining the converged estimate which yields the largest $\tilde{\ell}_{\text{LA}}(\btheta;\by_{1:T})$ value. Also, if the target tuning constant $c$ is low and relatively few outliers are expected in the data, it may help convergence to start with a larger $c$ value and decrease it sequentially, each time computing the corresponding uncorrected robust estimator and using it as a starting point in the next iteration. Ideally a highly robust starting point would be preferable, but initializing with the MLE in such a sequence regularly yielded reasonable outcomes in preliminary simulations, as it only has an indirect impact on the final robust estimate. Second, if the target efficiency (and thus $c$ value) is quite high, our experience is that the bias of the uncorrected $\tilde{\btheta}_T^{[1]}$ may be close to negligible for all practical purposes and thus may be useful on its own given its fast computation. The simulation results in Appendix~S4, where the robust estimator is tuned for 90\% efficiency show such a pattern for many parameters, see Section~\ref{s:sim}. Third, in the NR step we make use of the uncorrected Hessian matrix $\partial \tilde{\bs}(\btheta;\by_{1:T})/\partial \btheta$, which is approximated by finite differences. We have tested correcting it by Monte Carlo in a similar fashion as in (\ref{eq:MCapproxFcct}) but found little improvement in $\tilde{\btheta}$. Finally, Step~2 and 3 may be iterated. A very large $B$ (and fixing the seed value in Step~2) is necessary to ensure the iterated procedure converges, but this makes the whole excessively time-consuming. Our experience is that a single step seems to effectively offset any noticeable bias, see Section~\ref{ss:simres}.

\subsection{Theoretical properties}\label{ss:theory}

Here we study asymptotic properties (consistency and robustness) of our proposed robust estimator $\tilde{\btheta}_T$. For this, we approximate the uncorrected robust score $\tilde{\bs}(\btheta;\by_{1:T})$ by a sum of individual scores which form a stationary process. In the limit $T \rightarrow \infty$, the robust estimator corresponds then to a functional on the space of stationary and ergodic processes, see Theorem~\ref{th:consistency} below.

We require the following assumptions for our main result.
\begin{enumerate}[label=(A\arabic{*})]
\item{\label{A:rho} 
For all $c < \infty$, the $\rho_c$ function is convex, bijectively increasing and twice continuously differentiable over $\mathbb{R}$, and satisfies the following:
\begin{enumerate}[label=(\roman{*})]
  \item{\label{rhocond1} $\displaystyle \lim_{z \rightarrow -\infty} \rho_c(z)/z = 0$,
  }
  \item{\label{rhocond2} $\displaystyle \lim_{z \rightarrow +\infty} \rho_c(z)/z = 1$,
  }
  \item{\label{rhocond3} $\displaystyle \int_{\mathsf{Y}^T}\int_{\mathsf{X}^{(T+1)}} \tilde{p}_{\btheta}(\bx_{0:T}, \by_{1:T}) \diffd \bx_{0:T} \diffd \by_{1:T} < \infty$. 
  }
\end{enumerate}
}
\item{\label{A:compact} 
The state space $\mathsf{X}$ is compact.
}
\item{\label{A:poscont} 
The unobserved state density satisfies
\begin{equation*}
\inf_{\btheta \in \bTheta}\inf_{\bx, \bx' \in \mathsf{X}} f_{\btheta}(\bx,\bx') = \sigma_{-} > 0 \text{ and } \sup_{\btheta \in \bTheta}\sup_{\bx, \bx' \in \mathsf{X}} f_{\btheta}(\bx,\bx') = \sigma_{+} < \infty.
\end{equation*}
The observation density satisfies $0 < g_{\theta}(\bx,\by) < \infty$ for all $\bx \in \mathsf{X}$, $\by \in \mathsf{Y}$ and $\btheta \in \bTheta$, as well as
\begin{equation*}
\sup_{\btheta \in \bTheta}\sup_{\bx \in \mathsf{X}}\sup_{\by \in \mathsf{Y}} |\partial/\partial \btheta \log g_{\theta}(\bx,\by)| < \infty.
\end{equation*}
Also, both mappings $\btheta \rightarrow g_{\theta}(\bx,\by)$ and $\btheta \rightarrow f_{\theta}(\bx,\bx')$ are twice continuously differentiable with respect to $\btheta$.
}
\item{\label{A:compacttheta} 
The parameter space $\bTheta$ is compact.
}
\end{enumerate}

Assumptions~\ref{A:rho}\ref{rhocond1}--\ref{rhocond2} state that the weight represented by $\rho'$ is bounded within $[0,1]$, which is the main requirement for the robustness property. Assumption~\ref{A:rho}\ref{rhocond3} is not formally needed for robustness but ensures the normalization constant can be neglected by viewing $\tilde{p}_{\btheta}$ as a probability density function, see Appendix~S1 of the Web-based Supporting Information. The compactness of $\mathsf{X}$ in \ref{A:compact} in conjunction with \ref{A:poscont} ensures the existence of the invariant stationary distribution of $\bX$ while avoiding explicit positive recurrence requirements. We note that these are not as restrictive as requiring $\mathsf{X}$ to be a finite set, a special case of SSMs often referred to as hidden Markov models. The compactness of $\mathsf{X}$ may be relaxed following \citet{douc2011}, although at the cost of additional technical assumptions which go beyond the scope of the paper. In \ref{A:poscont}, the requirements on $f_{\theta}$ and $g_{\theta}$ are fairly standard \citep[e.g.][]{douc2004}. The uniform boundedness of $f_{\theta}$ implies that the state space $\mathsf{X}$ is 1-small thanks to \ref{A:compact} \citep[see][Section~5.2]{meyn1993} and thus that the process is uniformly ergodic by Theorem~16.0.2 of \citet{meyn1993}. This means that, regardless of the starting point, the distribution of $\bX$ converges in the total variation norm to its invariant distribution at a uniform geometric rate. As a corollary, the joint process $(\bX,\bY)$ is also uniformly ergodic. In order to relax the uniform boundedness in \ref{A:poscont}, a possible route would be to explicitly assume that $\bX$ is Harris positive recurrent to obtain joint uniform ergodicity, see Lemma~1 of \citet{douc2011}. Finally, we note that the compactness of $\bTheta$ is standard for studying the estimation of SSMs \citep{douc2004,douc2011,douc2012} and cannot be easily relaxed.

Based on these assumptions, the following theorem establishes that the (normalized) corrected robustified score equations are asymptotically unbiased, with proof given in Appendix~S1 of the Web-based Supporting Information.

\begin{theorem}\label{th:consistency} 
Under \textnormal{\ref{A:rho}}--\textnormal{\ref{A:compacttheta}}, there exists a bounded measurable function $\bzeta: \bTheta \times \mathsf{Y}^{\mathbb{Z}} \rightarrow \mathbb{R}^p$, where $\mathbb{Z}$ denotes the set of all integers, such that
\begin{equation*}
\sup \lim_{T\rightarrow\infty}T^{-1} \left| \tilde{\bs}(\btheta;\by_{1:T}) - \sum_{t=1}^T \bzeta(\btheta; \by_t,\by_{t\pm 1},\by_{t \pm 2}, \ldots) \right| = \boldsymbol{0},
\end{equation*}
where the supremum is taken over all $\btheta \in \bTheta$ and over all doubly-infinite sequences $\by_{-\infty:+\infty}$. In particular, for any stationary and ergodic process $(\bY_t)_{t \in \mathbb{Z}}$, $T^{-1} \tilde{\bs}(\btheta;\by_{1:T})$ converges almost surely and in expectation to $\Expec[\bzeta(\btheta;\bY_0,\bY_{\pm 1},\bY_{\pm 2},\ldots)]$. If we set
\begin{equation*}
\bzeta^\star(\btheta;\bY_0,\bY_{\pm 1},\bY_{\pm 2},\ldots) = \bzeta(\btheta;\bY_0,\bY_{\pm 1},\bY_{\pm 2},\ldots) - \lim_{T\rightarrow\infty}T^{-1}b_T(\btheta),
\end{equation*}
then the statistical functional $S^\star(\Plaw)$ defined as the solution in $\btheta$ to
\begin{equation*}
\int \bzeta^\star(\btheta;\by_0,\by_{\pm 1},\by_{\pm 2},\ldots) \, \Plaw(\text{d} \by) = \boldsymbol{0}
\end{equation*}
satisfies by construction the condition for Fisher consistency, namely that $S^\star(\Plaw_{\btheta}) = \btheta$ for all $\btheta \in \bTheta$.
\end{theorem}

The function $\bzeta$ in Theorem~\ref{th:consistency} is not meant to be computed in practice, it is a theoretical tool from which properties like consistency and robustness (see Theorem~\ref{th:IF} below) can be established. Also, Theorem~\ref{th:consistency} deals with the idealized score equations $\tilde{\bs}^\star(\btheta;\by_{1:T}) = \boldsymbol{0}$. It would be of interest to extend it to cover also the effects of the Laplace approximation and of the one step bias correction proposed in Section~\ref{ss:implement}. Unfortunately, this seems very difficult, and the asymptotic behavior of the error of the Laplace approximation is unclear in a general state space framework. This challenge already arises in ML estimation, see \citet[][Section~4.1]{rue2009} where the issue remains unresolved in a simpler setting. That said, we ran many simulations (not presented here) under various designs to assess this: we can confirm that the Laplace approximation error appears practically negligible in terms of parameter estimates when it is applied on a suitable scale, i.e.\ when the joint likelihood is unimodal.

To discuss the robustness properties of $\tilde{\btheta}_T$ and its corresponding $S$, we consider deviations from the assumed data generating process formalized by the general replacement model of \citet{martin1986}: 
\begin{align}\label{eq:grm}
\bY_t^\epsilon &= (1-\bB_t^\epsilon)\bY_t + \bB_t^\epsilon \bW_t,
\end{align}
where $\bY_t$ is distributed according to the nominal $\Plaw_{\btheta}$, $(\bB_t^\epsilon)_{t\geq1}$ is a 0--1 process with $\Pr(\bB_t^\epsilon = 1) = \epsilon$, and $\bW = (\bW_t)_{t \geq 1}$ is a stationary and ergodic contamination process with distribution $\Plaw_{\bW}$. The latter is left unspecified apart from the required stationarity and ergodicity, following \citet[][Section~4]{martin1986}. The process $\bB_t^\epsilon$ is an indicator of contaminated observations with $\epsilon$ representing the fraction of observations that do not come from the assumed model. The general replacement model (\ref{eq:grm}) is flexible enough to consider outliers both isolated in time and occurring in patches, see \citet[][Section~2.2]{martin1986}. Deviations from the model can arise either from outliers in the state process $\bX$ which then propagate to the observations according to the nominal model with density $g$, or directly from outliers in the observations themselves.

We denote the distribution of the contaminated process $(\bY_t^\epsilon)_{t \geq 1}$ by $\Plaw_{\btheta}^{\epsilon}$. The influence functional (IF) of the statistical functional $S$ at the nominal model $\Plaw_{\btheta}$ is defined as
\begin{align*}
\text{IF}(\Plaw_{\bW},S,\Plaw_{\btheta}) &= \lim_{\epsilon \downarrow 0} \frac{S(\Plaw_{\btheta}^\epsilon) - S(\Plaw_{\btheta})}{\epsilon},
\end{align*}
provided the limit exists, see \citet{martin1986}. The IF here depends not only on the nominal model and the distribution of the contaminating process $\bW$, but also on whether the outliers are isolated or occur in patches. To see this, the following two assumptions are needed:
\begin{enumerate}[label=(A\arabic{*})]\setcounter{enumi}{4}
\item{\label{A:MYTh42a} 
$S^\star(\Plaw_{\btheta_0}^{\epsilon}) - \btheta_0 = O(\epsilon)$.
}
\item{\label{A:MYTh42b} 
The $(p \times p)$ matrix  $\bM(\btheta) = (\partial / \partial \btheta)\Expec[\bzeta^\star(\btheta;\bY_0,\bY_{\pm 1},\bY_{\pm 2},\ldots)]$ exists and is non-singular at $\btheta=\btheta_0$.
}
\end{enumerate}
These additional assumptions are among the rather mild technical conditions \citet{martin1986} require in their Theorem~4.2. Note that $\bM(\btheta)$ is nothing more than the Fisher information when $c \rightarrow \infty$. The theorem below states the main robustness result of our proposed estimator.

\begin{theorem}\label{th:IF} 
Under \textnormal{\ref{A:rho}}--\textnormal{\ref{A:MYTh42b}}, if $\lim_{\epsilon\downarrow 0} \Expec_{\Plaw_{\btheta}^{\epsilon}}[\bzeta^\star(\btheta;\bY^\epsilon_0,\bY^\epsilon_{\pm 1},\bY^\epsilon_{\pm 2},\ldots)]/\epsilon$ exists, then
$$ \textnormal{IF}(\Plaw_{\bW},S^\star,\Plaw_{\btheta}) = - \bM(\btheta)^{-1} \lim_{\epsilon\downarrow 0} \Expec_{\Plaw_{\btheta}^{\epsilon}}[\bzeta^\star(\btheta;\bY^\epsilon_0,\bY^\epsilon_{\pm 1},\bY^\epsilon_{\pm 2},\ldots)]/\epsilon.$$
Furthermore, for patchy outliers with arbitrary but fixed patch length $k$, the limit
$$\lim_{\epsilon\downarrow 0} \Expec_{\Plaw_{\btheta}^{\epsilon}}[\bzeta^\star(\btheta;\bY_0^\epsilon,\bY^\epsilon_{\pm 1},\bY^\epsilon_{\pm 2},\ldots)]/\epsilon$$
exists and we have the particular representation
$$ \textnormal{IF}(\Plaw_{\bW},S^\star,\Plaw_{\btheta}) = - \frac{1}{k}\bM(\btheta)^{-1} \sum_{j=-\infty}^{\infty} \Expec\left[ \bzeta^\star(\btheta;\bY_{-\infty:(j-k)}^0, \bW_{(j-k+1):j} , \bY_{(j+1):\infty}^0) \right],$$
where the latter expectation is taken under the joint distribution of $\bY$ and $\bW$.
\end{theorem}

The proof is deferred to Appendix~S2. In the case of patchy outliers with arbitrary but fixed patch length, the limit is guaranteed to exist and is uniformly bounded over all contamination distributions $\Plaw_{\bW}$, for any $c < \infty$. The representation we provide indicates that we only need to consider one patch of outliers at an arbitrary position which then has a bounded total effect on $\tilde{\bs}^\star(\btheta;\by_{1:T})$ as $T \rightarrow \infty$. The boundedness of the IF ensures that any bias caused by the arbitrary contaminating process $\bW$ is finite. In contrast, the MLE has an unbounded IF (as $c \rightarrow \infty$) meaning that even an infinitesimal $\epsilon$ can draw the estimate to arbitrarily large values \citep[see][Chapter~2.1]{hampel1986}. The MLE is thus, without much surprise, not robust.

We note that instead of the general replacement model (\ref{eq:grm}), one could also look at innovation outliers \citep{denby1979} in the state process, i.e.\ letting the contaminated state transition density of $\bX_t^\epsilon | \bx_{t-1}^\epsilon$ be $\big((1-\bB_t^\epsilon)f_{\btheta_0}(\bx_{t-1}^\epsilon, \bx_t) + \bB_t^\epsilon f_{\bW}(\bx_{t-1}^\epsilon, \bx_t)\big) \diffd \bx_t$ with arbitrary contamination density $f_{\bW}$. Although the conditional distribution of the observations $\bY^\epsilon$ given the states $\bX^\epsilon$ would still be according to the nominal model, outliers in the state sequence would have an effect on several observations. However, if the uncontaminated state process is geometrically ergodic, the effect of an outlier then disappears quickly. We would thus expect a similar behavior of the influence functional under our assumptions.

\section{Simulation study}\label{s:sim}

We investigate here the finite-sample performance of our proposed robust estimator through a simulation study. The model and design are based on the state space stock assessment model used for pollock in the North Sea. The original specification of this stock assessment model is denoted as ``Model~D'' in \citet{nielsen2014}, which is a version of the state space model known as SAM in the fisheries science literature. An important feature of this original specification is that the assumed unobserved state process is not stationary. Since our consistency and robustness properties rely on stationarity (see Section~\ref{ss:theory}), we thus define two versions of this SSM: a non-stationary version that closely follows \citet{nielsen2014} and a stationary one where the dynamics of $\bX$ are modified as little as possible and the observation equations remain unchanged. We note that we only introduce the stationary version to illustrate the properties of the robust estimator under controlled conditions. The non-stationary model is likely more realistic from a marine ecology point of view, hence we present results from both versions here and in Section~\ref{s:data}.

\subsection{Non-Stationary Model}\label{ss:modeloverview}

We detail here the non-stationary version of the SSM. The details of the stationary version are given in Appendix~S3.

All variables are indexed by the discrete fish age class $a$ and the year of observation $t$. The unobserved state vector $\bX_t$ consists of the unitless fishing mortality rate $F_{a,t}$ and the abundance of fish $N_{a,t}$ (as counts), both on the natural logarithm scale. The vector of observations $\bY_t$ comprises the total annual commercial catches $C_{a,t}$ (as counts) and a unitless index of relative abundance coming from yearly surveys $I_{a,t}$, also both considered on the log scale. The log transformations here are not based on some Box-Cox argument but rather come from the traditional interpretation in fish stock assessments of multiplicative nonnegative error on the original scale. The age ranges from $a=3$ years, so-called recruits, up to $a=A$ where the oldest age class includes fish of age $\geq A$ (a ``plus-group'') with the cut-off age $A$ varying from one variable to another: $A=9+$ for the fishing mortality $F$, $A=10+$ for the abundance $N$ and commercial catch $C$, and $A=8+$ for the survey index $I$. Thus the vector of unobserved states at year $t$ is
\begin{align*}
\bX_t &= (\log F_{3,t},\ldots,\log F_{9+,t},\log N_{3,t}, \ldots, \log N_{10+,t})^\top,
\end{align*}
with dimension $q=15$, and the vector of observations at year $t$ is
\begin{align*}
\bY_t &= (\log C_{3,t},\ldots,\log C_{10+,t},\log I_{3,t},\ldots,\log I_{8+,t})^\top,
\end{align*}
with dimension $r=14$. The time range is $T=50$, with $t=1$ corresponding to the year 1967 and $t=T$ corresponding to 2016. In the North Sea pollock data, $C_{a,t}$ is available up to year 2015 and $I_{a,t}$ spans the years 1992--2016, resulting in a sample size of $542$ observations. We replicate this missingness pattern in the simulated samples, noting there is sufficient overlap between the two data sources to cover the whole 1967--2016 time window.

The non-stationary version of the SSM is slightly different from the original specification of \citet{nielsen2014} in that we allow for a different process variance for the survival of the ``plus-group'' ($\sigma^2_P$) and we consider a single observation variance across all ages for the commercial catches ($\sigma^2_C$) and for the survey indices ($\sigma^2_I$), this for parsimony reasons.

The fishing mortalities $F_{a,t}$ follow a (vector) random walk on the natural logarithm scale:
\begin{align*}
\left[\begin{array}{c} \log F_{3,t}\\ \vdots \\ \log F_{9+,t}  \end{array}\right] &= \left[\begin{array}{c} \log F_{3,t-1}\\ \vdots \\ \log F_{9+,t-1}  \end{array}\right] + \bxi_t, \quad \forall t=2,\ldots,T,
\end{align*}
where the $\bxi_t$s are i.i.d.\ multivariate normal with zero mean vector and covariance matrix $\bSigma$ describing dependence across age classes akin to a first-order autoregressive process AR(1): 
$$ \bSigma_{a,a'} = \rho^{|a-a'|}\cdot\sigma_{F_a}\sigma_{F_{a'}}$$
for ages $a$ and $a'$. A distinct variance parameter is set for the first age class $\sigma^2_{F_{a=3}}$, while all other age classes share the same process error variance $\sigma^2_{F_{a\geq 4}}$. The dynamics of the true abundance $N_{a,t}$ vary between ages. For the youngest age class $a=3$ we have a random walk on the log scale
\begin{align*}
\log N_{3,t} &= \log N_{3,t-1} + \eta_{3,t},
\end{align*}
where $\eta_{3,t}$ are i.i.d.\ N$(0,\sigma^2_R) \; \forall t\geq2$. Then, for $4\leq a \leq 9$, the equation
\begin{align*}
\log N_{a,t} &= \log N_{a-1,t-1} - F_{a-1,t-1} - M_{a-1,t-1} + \eta_{a,t},
\end{align*}
specifies the survival of fish of intermediate ages by tracking cohorts through time, where $M_{a,t}$ denotes the natural mortality rate (all other causes except fishing) and $\eta_{a,t}$ are i.i.d.\ N$(0,\sigma^2_N) \; \forall t\geq2$. Natural mortality rates are obtained from other data sources such as stomach samples, and are considered a fixed covariate here (set to 0.2 for all ages and years). Finally, for the largest age class $a=10+$ we consider the additional survival of fish in the oldest ``plus-group'' from the previous year:
\begin{align*}
\log &N_{10+,t} =\\
&\log\Big[ N_{9,t-1} \exp\big\{-F_{9+,t-1} - M_{9,t-1}\big\} + N_{10+,t-1}\exp\big\{-F_{9+,t-1} - M_{10+,t-1}\big\}\Big]  + \eta_{10+,t},
\end{align*}
where $\eta_{10+,t}$ are i.i.d.\ N$(0,\sigma^2_P) \; \forall t\geq2$. As in the original specification of \citet{nielsen2014}, no distribution is explicitly specified for the initial states. This can be thought of as $h_{\btheta}$ being an extremely wide uniform density, with the dynamics mainly driving the state predictions.

Regarding observed variables, the \citet{baranov1918} catch equation for the total annual commercial catch $C_{a,t}$ is adapted as
\begin{align*}
\log C_{a,t} &= \log\left[ \frac{F_{a,t}}{Z_{a,t}} \big\{1-\exp(-Z_{a,t})\big\} N_{a,t} \right] + \epsilon^{C}_{a,t},
\end{align*}
where $Z_{a,t} = F_{a,t} + M_{a,t}$ is the total mortality rate and $\epsilon^{C}_{a,t}$ are i.i.d.\ N$(0,\sigma^2_{C})$, for $a=3,\ldots,10+$ and $t=1,\ldots,T-1$. Commercial catch data is indeed only available up to year 2015. Finally, the survey indices of relative abundance $I_{a,t}$ are linked to the total mortality and true abundance through
\begin{align*}
\log I_{a,t} &= \log\left[ q_{a} \cdot \exp\left\{ -Z_{a,t}\cdot\frac{\#\text{days}}{365} \right\}\cdot N_{a,t} \right] + \epsilon^{I}_{a,t},
\end{align*}
where $q_a$ is known as a catchability coefficient as it scales the relative index to the true abundance by how likely a fish of a given age can be caught, \#days stands for the number of days into the year the survey has been conducted ($\#\text{days}/365 = 0.622$ here), and $\epsilon^{I}_{a,t}$ are i.i.d.\ N$(0,\sigma^2_{I})$. The catchability coefficients are different for all six age classes $a=3,\ldots,8+$, while the time index $t$ only spans the years 1992--2016.

Overall, all error terms are Gaussian on the log scale, but relations are highly non-linear. The $p$-dimensional parameter to estimate is
\begin{align*}
\btheta = (\sigma_{F_{a=3}}, \sigma_{F_{a\geq4}}, \rho, \sigma_R, \sigma_N, \sigma_P, \sigma_C, q_{a=3}, \ldots, q_{a=8+}, \sigma_I)^\top,
\end{align*}
with $p=14$. Regarding the stationary version of the model, random walks are replaced by first-order autoregressive processes with non-zero stationary expectations, and some survival equations are modified so as to derive a stationary distribution (log-normal on the original scale). In the stationary version more parameters need to be estimated, such as the fishing mortality stationary mean for each age class, so that $p=26$ in this case; see Appendix~S3 for details.

\subsection{Simulation design}\label{ss:simdesign}

For both versions of the SSM, we simulate $500$ replicates both at the assumed model and under contamination (see below), and compare our robust estimator to the MLE in each scenario. On the one hand, the goal at the model is to ensure that the robust estimator is not biased relative to the MLE, i.e.\ the Fisher consistency correction step is effective. On the other hand, the purpose of contaminating the data is to study the reliability of the robust estimator in the presence of deviations from model assumptions. To simulate samples at the model, we use the robust estimates of $\btheta$ from the North Sea pollock data. Table~S1 in Appendix~S4 reports these values along with the starting values we used for the MLE. Regarding the initial states, since the non-stationary version of the model leaves them free, we fix them in this case to the estimated values from an initial run of Model~D on the North Sea pollock data (accessible at stockassessment.org). These fixed initial states are only used for simulating data, they are not part of the model fitted to data, both here and in Section~\ref{s:data}. In the case of the stationary version, we simulate the initial states from the stationary invariant distribution.

Our contamination scheme is meant to somewhat mimic what we observe in the North Sea pollock data: we select a single year, the year 2000, and subtract a constant value (of 4) from the realized log-catches of the recruits $\log C_{a=3,t}$, while enforcing a floor at 1 to retain realistic values. Hence we contaminate only part of the observation vector $\bY_t$, leaving the survey indices untouched. By modifying only a single age class in a single year the overall contamination proportion is $1/542\approxeq 0.2\%$, taking into account the missing values in $C_{a,t}$ and $I_{a,t}$. With simulated $\log C_{a=3,t}$ values typically around 8 surrounding the year 2000, subtracting 4 is sufficient to create a noticeable bias in the MLE. We note that the contamination proportion could be increased, affecting other components of the observation vector, but only up to a certain point: the breakdown point of our proposed robust estimator is likely no better than the $1/p$ benchmark (i.e.\ $1/14 \approxeq 7.1\%$ for the non-stationary version of the model) set by $M$-estimators in simpler models such as linear regression.

We use the proposed SSH $\rho$ function which satisfies \ref{A:rho}, while \ref{A:compact}--\ref{A:MYTh42b} are satisfied by the stationary version of the model if we concede that the additive Gaussian error terms for $\bX$ are supported on a compact set, a concession which is often (implicitly) made in the literature, e.g.\ in non-parametric statistics. A compact state space is arguably realistic here since both the fish stock abundance $N_{a,t}$ and the fishing mortality $F_{a,t}$ are bounded by finite natural resources. We could make this explicit in the model by considering a truncated Gaussian distribution, but without information about where to truncate we would typically truncate far out in the tails (say, beyond $\pm 10$ for a standard Gaussian). So the results would end being practically identical to those presented here with a non-truncated Gaussian distribution. We tune the robust estimator to achieve around $90\%$ empirical efficiency based on the MAD for all likelihood contributions. Due to the various scale parameters, we set six different tuning constants, one for each of the following likelihood contributions: $(\log F_{3,t},\ldots,\log F_{9+,t})^\top$; $\log N_{3,t}$; $\log N_{a,t}$ for $4\leq a \leq 9$; $\log N_{10+,t}$; $(\log C_{3,t},\ldots,\log C_{10+,t})^\top$; and $(\log I_{3,t},\ldots,\log I_{8+,t})^\top$. The resulting $c$ values are $(1.2, 2.3, 1.0, 1.3, 1.3, 2.9)$ and $(2.0, 2.2, 1.5, 1.5, 1.2, 2.9)$ for the non-stationary and stationary model, respectively.

The MLE serves as a starting point and the optimization in Step~1 is done in two steps with decreasing $c$ (see Section~\ref{ss:implement}). We use $B=1000$ replicates for the Monte Carlo approximation of the Fisher consistency correction term.

Our implementation and all simulations are performed in R version 3.4.3 \citep{R2017} with TMB version 1.7.12. Documented code, including TMB templates, is readily available on the first author's GitHub page.

\subsection{Simulation results}\label{ss:simres}

Figure~\ref{fig:boxplots} presents boxplots of ML and robust estimates of selected parameters, both at the assumed model and under contamination. For each version of the model, the two selected parameters are the ones for which the MLE is the least and most biased under contamination. This serves as an illustration, while Figures~S1 and S2 in Appendix~S4 display similar boxplots for all elements of $\btheta$, respectively for the non-stationary model and for the stationary version, with the addition of the uncorrected robust estimator $\tilde{\btheta}_T^{[1]}$ computed at the model to show how effective the one-step correction in Step~3 is. At the assumed model, both sets of estimates are well-centered, with the exception perhaps of $\phi_F$ and $\phi_R$ in the stationary version whose slight underestimation is likely due to the small sample size ($T=50$). The robust estimates are generally more variable than the MLE, as is typically expected. The uncorrected robust estimator shows some slight bias only for a few parameters, this being because of the high efficiency aimed for, resulting in rather high tuning constant values. Under contamination, not all parameters are biased when estimated by ML, this being due to our particular contamination scheme affecting only $\log C_{a=3,t}$ at a single time point. But the robust estimates are systematically closer to the true parameter values and with variability often similar to that of estimates at the model.

\begin{figure}[ht!]
\begin{center}
\includegraphics[width=\columnwidth]{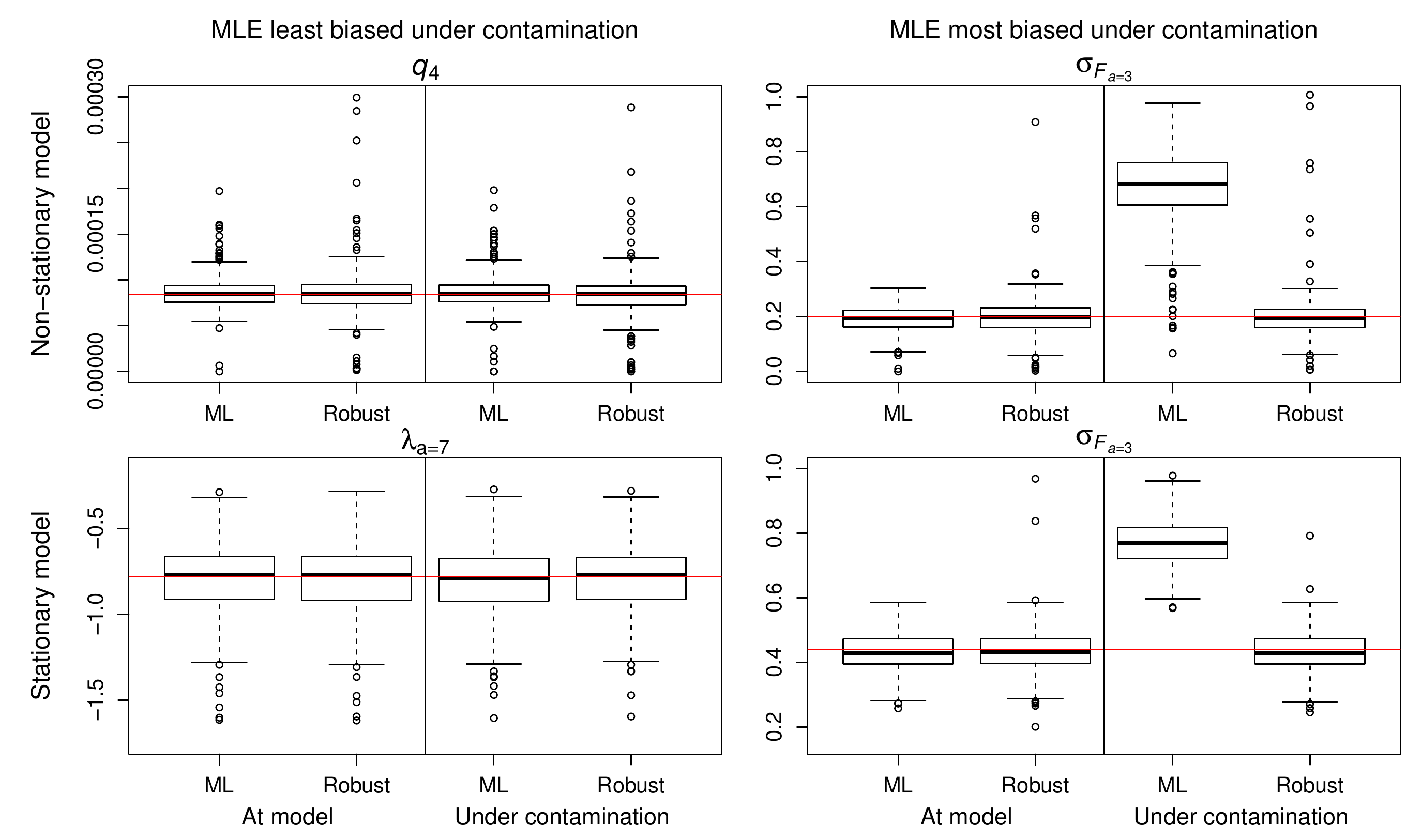}
\end{center}
\caption{Boxplots of ML and robust estimates for those parameters where the MLE is least/most biased under contamination, for the non-stationary and stationary versions of the model, at the assumed model and under contamination. The red horizontal solid line is the true parameter value. The vertical scale is manually set for better visualization, some points (both ML and Robust) not shown.
{\label{fig:boxplots}}%
}
\end{figure}

Figure~\ref{fig:boxplotsRMSEnst} below and Figure~S3 in Appendix~S4 present boxplots of root MSE (RMSE) of in-sample predictions of the unobserved states $\bX_t$ as computed from the Laplace approximation for the non-stationary and the stationary version of the model, respectively. The ML predictions (left column) are based on the maximization of the SSM likelihood in (\ref{eq:jointlkhd}), akin to posterior modes, while the robust ones (right column) are obtained by maximizing the robustified likelihood in (\ref{eq:robjointlkhd}), both given their respective estimates for $\btheta$. The RMSEs are computed by averaging squared differences over the entire $\bX_t$ vector for $t=1,\ldots,T$, so that each boxplot corresponds to a given year and represents the $500$ replicates. At the assumed model, the robust predictions tend to vary slightly more than the ML ones, especially for the non-stationary version. This is expected given the general loss of efficiency, but it could also be due to the fact that the posterior modes derived from $\tilde{p}_{\btheta}(\bx_{0:T} \,|\, \by_{1:T})$ do not coincide in general with those coming from the nominal $p_{\btheta}(\bx_{0:T}, \by_{1:T})$. Correcting the robust state predictions is out of scope of the paper, as we focus on the estimation of $\btheta$, but we note that any prediction bias due to not correcting seems quite marginal in our simulations. Comparing with the contamination scenario, it is clear that ML-based predictions get heavily biased around the contaminated year 2000 (grey-colored box). The Markovian dynamics of $\bX$ are likely responsible for how local this bias appears. By contrast, the RMSEs of robust predictions are comparable between the two scenarios, hinting that the unobserved states can be predicted in a robust way along with the estimates for $\btheta$.

\begin{figure}[ht!]
\begin{center}
\includegraphics[width=\columnwidth]{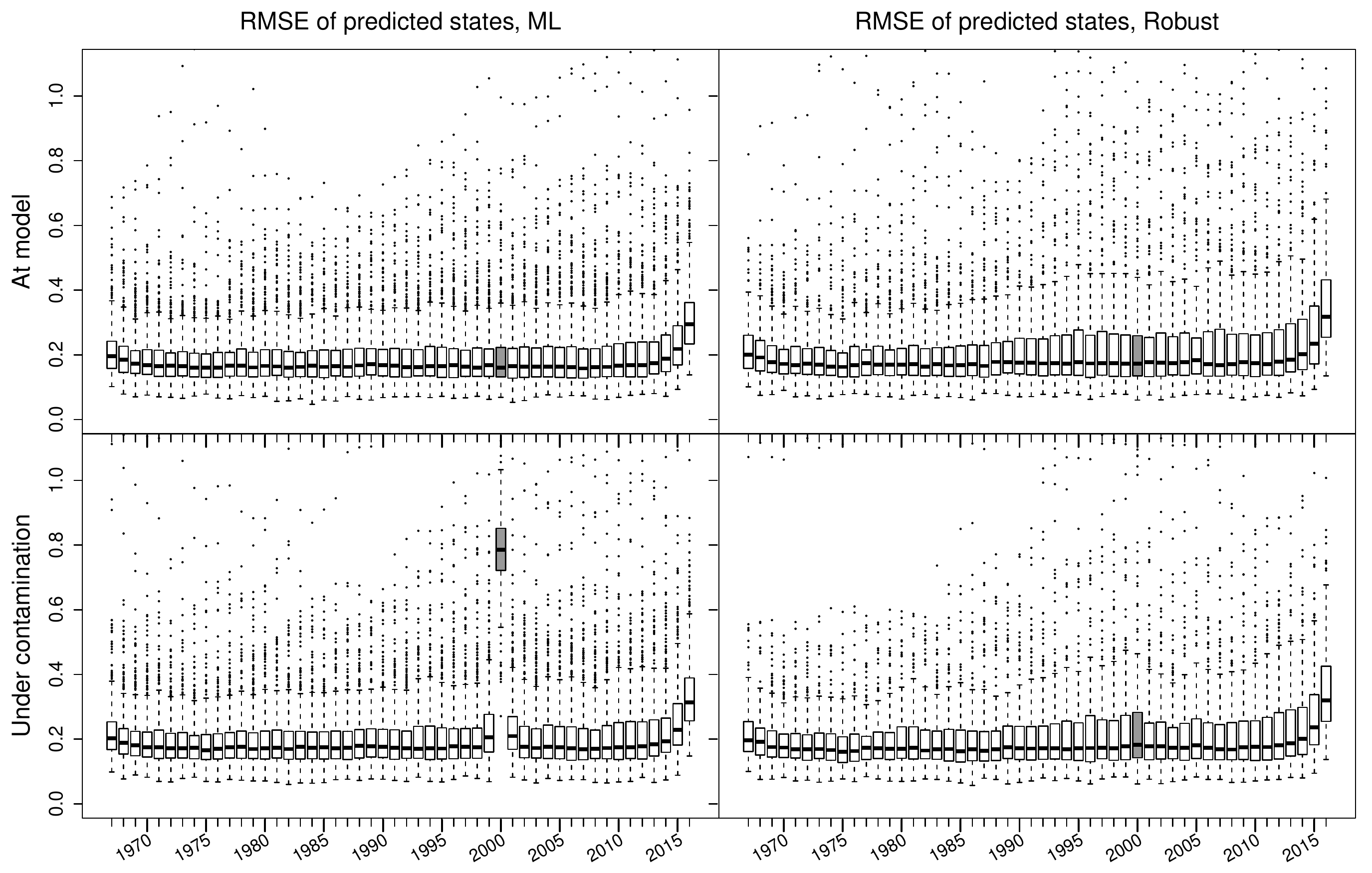}
\end{center}
\caption{Boxplots of RMSE of predicted states $\bX_t$ for the non-stationary version of the model, based on ML (left column) and robust (right column) methods, at the assumed model (top row) and under contamination (bottom row). The box corresponding to the contaminated year 2000 is colored in grey. The vertical scale is manually set for better visualization, some points (both ML and Robust) are not shown.
{\label{fig:boxplotsRMSEnst}}
}
\end{figure}

The robustness weights represented by $\rho'$ are also worth investigating. At the model, under the non-stationary version, a weight of 1 happens for at least 98\% of all likelihood contributions, replicates and time points, while the 1\% empirical quantile is 0.979. Under the stationary version, the 1\% quantile is exactly 1. In other words, at the model the robustness weights are nearly always 1, meaning that the $\rho$ function rarely affects estimation. Under contamination, the weights on the contributions of $\log F_{a,t}$, $\log N_{a,t}$ and $\log I_{a,t}$ are still very high: for the non-stationary model, at least 97\% over all replications and time points equal 1, while the 2\% quantile is 0.966. For the stationary version, the 2\% and 1\% quantiles are 1 and 0.997, respectively. However, the weights on $\log C_{a,t}$ show a distinct pattern under contamination. Table~S2 in Appendix~S4 presents quartiles of the weights for the recruits ($a=3$) over three consecutive years. The contaminated year 2000 shows systematically low weights while nearby years do not, nor do the same years at the model. This is clear evidence that the robust method is able to detect the deviating observations. By effectively downweighting them, we obtain reliable estimates and predictions as presented above.

\section{North Sea pollock stock assessment}\label{s:data}

We now apply the proposed robust estimator to the assessment of the pollock fishery in the North Sea. The pollock stock covers Division IIIa and Subareas IV and VI as defined by ICES, see Figure~S4 in Appendix~S5. The yearly total commercial catches $C_{a,t}$ and survey indices of relative abundance $I_{a,t}$ are only available for selected years within the 1967--2016 window, see Appendix~S3 for details. Additional information about the pollock stock can be found in \citet[][p.~495]{ices2015}. We again stress that the non-stationary (original) SSM is a real-life model used for the assessment of fish stocks and directly contributes to the science-based management of fisheries.

We fit both versions of the SSM to the North Sea pollock data and compare outputs from the ML and robust methods. The SSH $\rho$ function tuning constant values are the same as in the simulations of Section~\ref{s:sim} and we generate $B=5000$ samples for the Monte Carlo approximation of $b_T(\btheta)$. Table~S3 in Appendix~S5 reports computation times with a laptop housing a 2.9 GHz CPU: while the simulation-based correction step for the robust estimator can be lengthy depending on the value of $B$, the optimizations over $\btheta$ take at most a few seconds with these data. Table~S4 in Appendix~S5 reports the estimates and standard errors for both methods and both model versions. Standard errors for the robust estimates are based on the uncorrected estimates and are solely given for reference. Overall, the estimates are qualitatively comparable for the two estimation methods. This means that no significant deviations from model assumptions seem to draw the MLE to arbitrary values and that such estimates may thus be trusted.

The state predictions show a more nuanced message, they are presented in Figure~\ref{fig:dataprednst} below for the non-stationary version of the model and in Figure~S5 in Appendix~S5 for the stationary version. The colored curves represent the point-wise predictions for the two components of $\bX$, the fishing mortality rate $F_{a,t}$ and the abundance $N_{a,t}$, while the shaded envelopes are constructed as $[\text{prediction} \pm 2 \times\text{standard error}]$ on the log scale with the bounds then exponentiated back to the original scale of $F_{a,t}$ and $N_{a,t}$. These envelopes are not intended for inference but are displayed to give a sense of the variability around predictions. While in general the ML and robust predictions tend to agree, there are a few instances where they differ markedly under the non-stationary model (indicated by red circles in the figure): before 1970 the robust predictions for $F_{a,t}$ are rather constant while the corresponding ML predictions feature a trough in 1968; in the mid-1970s the peak in recruits is much sharper for ML predictions than for robust ones; also in the mid-1970s, fishing mortalities are predicted at higher levels according to the ML method. While these differences may seem subtle, they can have important consequences since a strikingly abundant cohort gives the impression the fish stock is in good health and can withstand further fishing pressure. Under the stationary model, predictions are quite similar between methods but standard errors are often larger for robust estimates, hinting at less confidence in the predicted trajectories.

\begin{figure}[ht!]
\begin{center}
\includegraphics[width=\columnwidth]{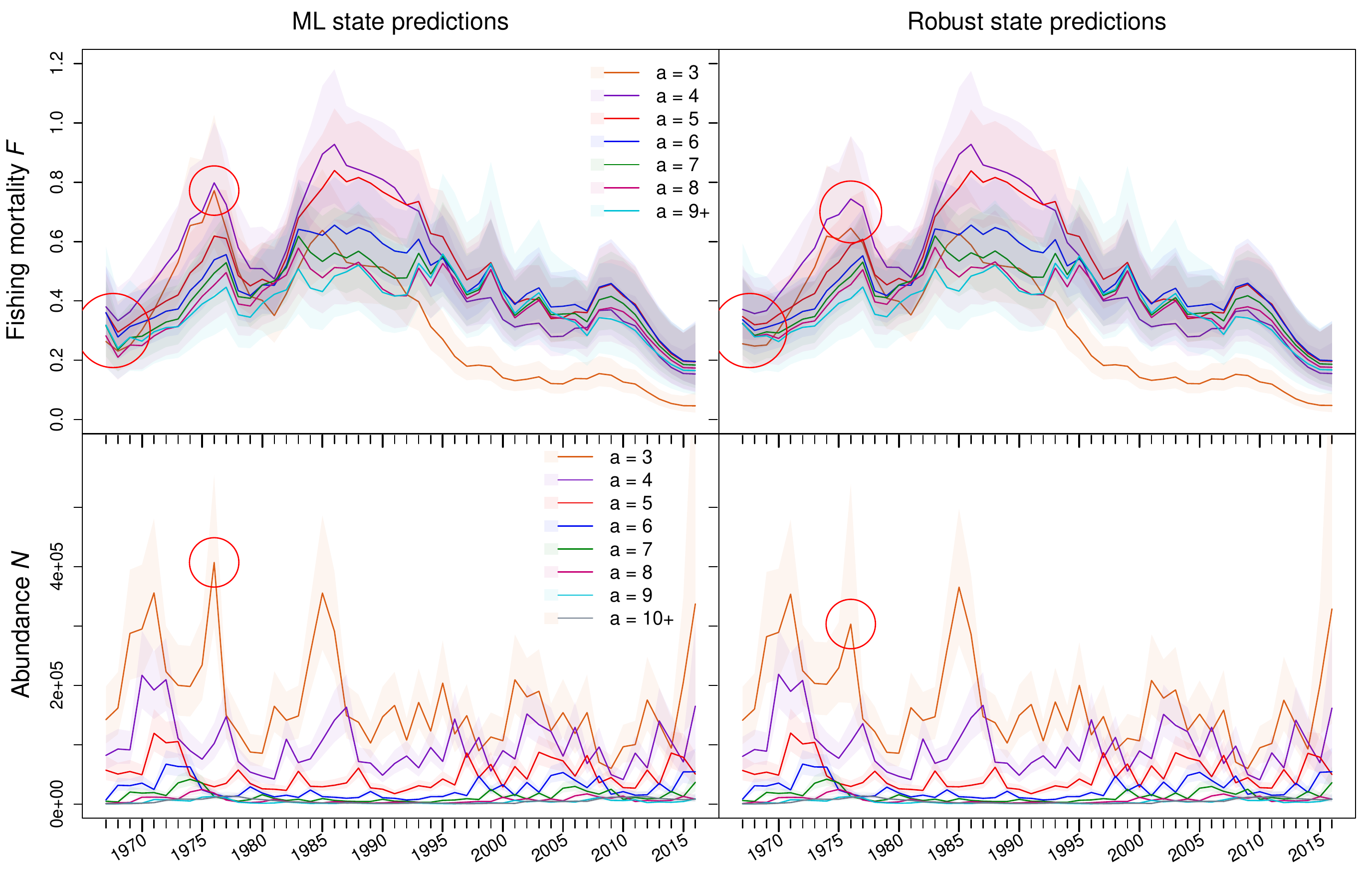}
\end{center}
\caption{North Sea pollock data ML and robust predictions of $\bX$ (on the original scale) for the non-stationary version of the model. The colored envelope are interpolated point-wise intervals constructed as $[\text{prediction} \pm 2 \times\text{standard error}]$ on the log scale with the bounds then exponentiated back to the original scale of $F_{a,t}$ and $N_{a,t}$. The circles indicate where the predictions differ markedly between the two methods.
{\label{fig:dataprednst}}%
}
\end{figure}

The robustness weights can shed light on the discrepancies between model outputs. The weights pertaining to the unobserved states are all virtually equal to 1, the ones corresponding to the observations are shown in the top and middle rows of Figure~S6 in Appendix~S5. Three data points are heavily downweighted under the non-stationary model, they are circled in the two bottom panels where the observed log catches time series are shown for all age classes within selected time windows. Two of these downweighted data points happen in 1968 and seem unusually low given the local trends. These two points are responsible for the ``kink'' in ML predictions for $F_{a,t}$ prior to 1970 in Figure~\ref{fig:dataprednst}. The robust method, by downweighting these two points, thus yields predictions that are more stable around those years. The third downweighted observation happens in 1976, it appears as a sudden peak in the recruits while other age classes do not show such a peak. The fact that this abundant cohort cannot be traced through the catches of older fish in the following years casts doubt on the validity of this particular data point. The robust method downweights it and this results in the much subdued peak of predicted abundance as seen in Figure~\ref{fig:dataprednst}. 

It is commonly known that commercial catches were not reported and recorded in a reliable way until the early 1980s for many fish stocks in the North Sea. Under the non-stationary model, the robustness weights objectively identify a few observations that seem to confirm this narrative. We note that the same observations do not receive particularly low weights under the stationary model. This means that this different model is somehow flexible enough to accommodate these data points. Observations are indeed only outlying given a model and its specific assumptions.

\section{Discussion and future directions}\label{s:conc}

We have introduced a robust and consistent estimation method for SSMs, including both a novel SSH $\rho$ function that does not require a compact observation outcome space and a simulation-based procedure to tune the robustness-efficiency tradeoff. The setting is quite general, requiring mainly that densities are available analytically. In addition, our proposed implementation utilizes state-of-the-art AD and the Laplace approximation to balance generality, computational cost and ease of programming. The motivating North Sea pollock fishery assessment illustrates how the robust method can identify atypical observations and limit their impact, guaranteeing reliable conclusions. Our simulations and application show that not only do parameter estimates remain reliable with contaminated data, but in-sample state predictions arising from the Laplace approximation (as posterior modes) also benefit from our robustification. Further work is needed to assess whether this is simply a byproduct of robust estimates of $\btheta$ or if the proposed method may yield robust state predictions for any given $\btheta$. Whether this would extend to robust filtering, as in \citet{calvet2015}, and to out-of-sample forecasting is an also open question. In addition, combining our robust estimation method with a robust filter represents important future research.

The simulation-based tuning of our proposed robust method is somewhat laborious due to the application of $\rho$ on the log-likelihood scale. The efficiency-robustness tradeoff will invariably depend on the model and design, but some standardization of the log-likelihood contributions could simplify the choice of the $c$ values. Also, the consistency and robustness properties we studied rely on the exact evaluation of the robust scores $\tilde{\bs}(\btheta;\by_{1:T})$. We plan to further study how the Laplace approximation involved in our proposed implementation may impact such properties. We note that such a study is in fact lacking in the literature even for ML estimation. Finally, we observed that our simulations show similar performance of the robust estimator for both non-stationary and stationary SSMs. This encourages the extension of some results to non-stationary processes, allowing a certain broader applicability of the method.

\section*{Acknowledgments}


This research was carried out as part of the collaborative research team initiative Advancements to State Space Models for Fisheries Science funded by the Canadian Statistical Sciences Institute.




\renewcommand{\thesection}{\Alph{section}}
\setcounter{equation}{0}
\setcounter{section}{0}

\pagebreak

\begin{center}
\Large{Web Appendices for \textit{Robust Estimation for Discrete-Time State Space Models}}

\bigskip
\normalsize{William H.\ Aeberhard$^{1,2}$, Eva Cantoni$^{3}$, Chris Field$^{2}$,\\ Hans R.\ K\"{u}nsch$^{4}$, Joanna Mills Flemming$^{2}$, and Ximing Xu$^{5}$}\\[0.7ex]
\footnotesize{$^{1}$Department of Mathematical Sciences, Stevens Institute of Technology}\\[0.7ex]
\footnotesize$^{2}$Department of Mathematics and Statistics, Dalhousie University\\[0.7ex]
\footnotesize$^{3}$Research Center for Statistics and GSEM, University of Geneva\\[0.7ex]
\footnotesize$^{4}$Seminar f\"{u}r Statistik, ETH Zurich\\[0.7ex]
\footnotesize$^{5}$School of Statistics and Data Science, Nankai University\\
\end{center}

\section*{Appendix S1: Proof of Theorem~1}

The proof of Theorem~1 in the main body of the paper can be summarized as follows: first, we show that the uncorrected robustified marginal score $\tilde{\bs}(\btheta;\by_{1:T})$ can be written as a sum of random variables; second, we prove that the summands are approximately stationary and ergodic; third, using Birkhoff's ergodic theorem we show that $\tilde{\bs}(\btheta;\by_{1:T})$ converges to a non-random function; finally, we wrap-up the proof by showing that the correction term $b_T(\btheta)$ guarantees unbiased estimating equations asymptotically.

We start by finding a more tractable expression for $\tilde{\bs}(\btheta;\by_{1:T})$. Let
\begin{align*}
\tilde{h}_{\btheta}(\bx) &= \exp\big[\rho_c\big\{\log h_{\btheta}(\bx)\big\}\big]\\
\tilde{f}_{\btheta}(\bx,\bx') &= \exp\big[\rho_c\big\{\log f_{\btheta}(\bx,\bx')\big\}\big]\\
\tilde{g}_{\btheta}(\bx,\by) &= \exp\big[\rho_c\big\{\log g_{\btheta}(\bx,\by)\big\}\big],
\end{align*}
so that 
\begin{align*}
\tilde{p}_{\btheta}(\bx_{0:T}, \by_{1:T}) &= \tilde{h}_{\btheta}(\bx_0) \prodtT \tilde{f}_{\btheta}(\bx_{t-1},\bx_t) \tilde{g}_{\btheta}(\bx_t,\by_t).
\end{align*}
The functions $\tilde{f}$, $\tilde{g}$ and $\tilde{h}$ are not proper densities since they do not integrate to one in general. That said, given that we will correct for Fisher consistency, we can ignore the normalization constants thanks to (A1)(iii) and will hereafter consider $\tilde{p}_{\btheta}(\bx_{0:T}, \by_{1:T})$ as if it is a density function. Note that the normalizing constant needed to make $\tilde{p}$ a proper density cancels out in the robustified conditional densities which appear below in Lemma~\ref{l:Step0FisherId} and \ref{l:Step1InhomogenousMarkovChain}. Let $\tilde{p}_{\btheta}(\by_{1:T}) = \int_{\mathsf{X}^{(T+1)}} \tilde{p}_{\btheta}(\bx_{0:T}, \by_{1:T}) \diffd \bx_{0:T}$. The following lemma yields a result analogous to the Fisher identity \citep[see e.g.][Appendix~D]{douc2014}.

\begin{lemma}\label{l:Step0FisherId}
Under \textnormal{(A1)},
$$ \tilde{\bs}(\btheta;\by_{1:T}) = \tilde{\Expec}_{\btheta}\left[ \chi(\bx_0,\btheta) | \by_{1:T}\right] + \sumtT\tilde{\Expec}_{\btheta}\left[\alpha(\bx_{t-1},\bx_t,\btheta) | \by_{1:T}\right] + \sumtT \tilde{\Expec}_{\btheta}\left[\beta(\bx_t,\by_t,\btheta) | \by_{1:T}\right] ,$$
where
$\tilde{\Expec}_{\btheta}[\cdot | \by_{1:T}]$ denotes the conditional expectation involving the robustified density of the states given the observations $\tilde{p}_{\btheta}(\bx_{0:T} | \by_{1:T}) = \tilde{p}_{\btheta}(\bx_{0:T}, \by_{1:T})/\tilde{p}_{\btheta}(\by_{1:T})$ and
\begin{align*}
\chi(\bx,\btheta) &= \frac{\partial}{\partial \btheta}\log\tilde{h}_{\btheta}(\bx),\\
\alpha(\bx,\bx',\btheta) &= \frac{\partial}{\partial \btheta}\log\tilde{f}_{\btheta}(\bx,\bx'),\\
\beta(\bx,\by,\btheta) &= \frac{\partial}{\partial \btheta}\log\tilde{g}_{\btheta}(\bx,\by).
\end{align*}
\end{lemma}

\begin{proof}
From the product chain rule for differentiation, we have
\begin{align*}
\frac{\partial \tilde{p}_{\btheta}(\bx_{0:T}, \by_{1:T})}{\partial \btheta} &=  \tilde{p}_{\btheta}(\bx_{0:T}, \by_{1:T})\frac{\partial \log \tilde{p}_{\btheta}(\bx_{0:T}, \by_{1:T})}{\partial \btheta}\\
&=  \tilde{p}_{\btheta}(\bx_{0:T}, \by_{1:T}) \Bigg\{ \frac{\partial}{\partial \btheta}\log\tilde{h}_{\btheta}(\bx_0) + \frac{\partial}{\partial \btheta}\sumtT\log\tilde{f}_{\btheta}(\bx_{t-1},\bx_t)\\
&\phantom{\mathrel{=} \tilde{p}_{\btheta}(\bx_{0:T}, \by_{1:T})\Bigg\{} + \frac{\partial}{\partial \btheta}\sumtT\log\tilde{g}_{\btheta}(\bx_t,\by_t) \Bigg\}.
\end{align*}
Given the definition of the conditional robustified density $\tilde{p}_{\btheta}(\bx_{0:T} | \by_{1:T})$, by exchanging the order of integration and differentiation we obtain
\begin{align*}
\tilde{\bs}(\btheta;\by_{1:T}) &= \frac{\partial}{\partial \btheta} \log\tilde{p}_{\btheta}(\by_{1:T})\nonumber\\
&= \big(\tilde{p}_{\btheta}(\by_{1:T})\big)^{-1}\frac{\partial}{\partial \btheta}\int_{\mathsf{X}^{(T+1)}} \tilde{p}_{\btheta}(\bx_{0:T}, \by_{1:T}) \diffd \bx_{0:T}\nonumber\\
&= \int_{\mathsf{X}^{(T+1)}} \frac{\partial \tilde{p}_{\btheta}(\bx_{0:T}, \by_{1:T})}{\partial \btheta} \frac{\tilde{p}_{\btheta}(\bx_{0:T} | \by_{1:T})}{\tilde{p}_{\btheta}(\bx_{0:T}, \by_{1:T})} \diffd \bx_{0:T}\nonumber\\
&= \int_{\mathsf{X}^{(T+1)}} \frac{\partial \log \tilde{p}_{\btheta}(\bx_{0:T}, \by_{1:T})}{\partial \btheta} \tilde{p}_{\btheta}(\bx_{0:T} | \by_{1:T}) \diffd \bx_{0:T}\nonumber\\
&= \tilde{\Expec}_{\btheta}\left[ \chi(\bx_0,\btheta) | \by_{1:T}\right] + \sumtT\tilde{\Expec}_{\btheta}\left[\alpha(\bx_{t-1},\bx_t,\btheta) | \by_{1:T}\right] + \sumtT \tilde{\Expec}_{\btheta}\left[\beta(\bx_t,\by_t,\btheta) | \by_{1:T}\right].
\end{align*}
\end{proof}

Lemma~\ref{l:Step0FisherId} shows that the uncorrected robustified marginal score can be written as a sum of random variables. We now turn to showing that each summand is approximately stationary. The following lemma introduces a necessary conditioning on an initial and an end point.

\begin{lemma}\label{l:Step1InhomogenousMarkovChain}
Under \textnormal{(A1)}, the conditional robustified density $\tilde{p}_{\btheta}(\bx_{m:n} | \bx_{m-1},\bx_{n+1}, \by_{m:n})$ corresponds to an inhomogeneous Markov chain for $0<m<n<T$.
\end{lemma}

\begin{proof}
The conditional robustified density of $\bx_{m:n}$ given $\bx_{0:(m-1)}$, $\bx_{(n+1):T}$ and $\by_{1:T}$ is proportional to $\tilde{f}_{\btheta}(\bx_{n},\bx_{n+1})\prod_{t=m}^n \tilde{f}_{\btheta}(\bx_{t-1},\bx_t)\tilde{g}_{\btheta}(\bx_t,\by_t)$ since the other terms in $\tilde{p}_{\btheta}(\bx_{0:T},\by_{0:T})$ do not contain $\bx_{n:m}$. As such it is independent of $\bx_{0:(m-1)}$, $\bx_{(n+2):T}$, $\by_{1:(m-1)}$ and $\by_{(n+1):T}$. In other words, thanks to the Markov property and conditional independence, it is sufficient to condition on contiguous $\bx$'s and contemporaneous $\by$'s. We thus denote this conditional robustified density by $\tilde{p}_{\btheta}(\bx_{m:n} | \bx_{m-1},\bx_{n+1}, \by_{m:n})$. Moreover, for $m \leq k \leq n$
\begin{align}\label{eq:forwarddensity}
\tilde{p}_{\btheta}&(\bx_k | \bx_{(m-1):(k-1)}, \bx_{n+1}, \by_{m:n})\nonumber\\
&= \int \tilde{p}_{\btheta}(\bx_{k:n} | \bx_{(m-1):(k-1)}, \bx_{n+1}, \by_{m:n}) \diffd \bx_{(k+1):n}\nonumber\\
&= \tilde{f}_{\btheta}(\bx_{k-1},\bx_k)\tilde{g}_{\btheta}(\bx_k,\by_k)\frac{\int \tilde{f}_{\btheta}(\bx_n,\bx_{n+1})\prod_{t=k+1}^n \tilde{f}_{\btheta}(\bx_{t-1},\bx_t) \tilde{g}_{\btheta}(\bx_t,\by_t) \diffd \bx_{(k+1):n}}{\int \tilde{f}_{\btheta}(\bx_n,\bx_{n+1})\prod_{t=k}^n \tilde{f}_{\btheta}(\bx_{t-1},\bx_t) \tilde{g}_{\btheta}(\bx_t,\by_t) \diffd \bx_{k:n}}.
\end{align}
Equation~(\ref{eq:forwarddensity}) provides the forward transition densities of an inhomogeneous Markov chain for $k>m$. These densities are independent of $\bx_{m-1}$ apart from the initial density which is given for $k=m$. Similarly we can compute the backward transition densities:
\begin{align}\label{eq:backwarddensity}
\tilde{p}_{\btheta}(\bx_k | \bx_{(k+1):(n+1)}, &\bx_{m-1}, \by_{m:n})\nonumber\\
&= \int \tilde{p}_{\btheta}(\bx_{m:k} | \bx_{(k+1):(n+1)}, \bx_{m-1}, \by_{m:n}) \diffd \bx_{m:(k-1)}\nonumber\\
&= \frac{\tilde{f}_{\btheta}(\bx_k,\bx_{k+1}) \int \prod_{t=m}^k \tilde{f}_{\btheta}(\bx_{t-1},\bx_t) \tilde{g}_{\btheta}(\bx_t,\by_t) \diffd \bx_{m:(k-1)}}{\int \tilde{f}_{\btheta}(\bx_k,\bx_{k+1})\prod_{t=m}^k \tilde{f}_{\btheta}(\bx_{t-1},\bx_t) \tilde{g}_{\btheta}(\bx_t,\by_t) \diffd \bx_{m:k}}.
\end{align}
\end{proof}

We now show that the inhomogeneous Markov chain in Lemma~\ref{l:Step1InhomogenousMarkovChain} ``forgets'' the values at both initial and end points as the lags increase.

\begin{lemma}\label{l:Step2ForgetTVnorm}
Under \textnormal{(A1)}--\textnormal{(A3)}, for $m \leq k \leq n$ and any $\bx_{m-1}$, $\bx_{n+1}$, and $\by_{m:n}$,
$$ \left\Vert\tilde{p}_{\btheta}(\bx_k | \bx_{m-1},\bx_{n+1}, \by_{m:n}) - \tilde{p}_{\btheta}(\bx_k | \bx_{m-1}', \bx_{n+1}', \by_{m:n}) \right\Vert_{\textnormal{TV}} \leq 2\max(\tau^{k-m+1},\tau^{n-k+1}),$$
where $||\cdot||_{\textnormal{TV}}$ denotes the total variation norm and $\tau = 1 - \tilde{\sigma}_{-}/\tilde{\sigma}_{+}$, where
$$\tilde{\sigma}_{-} = \inf_{\btheta \in \bTheta}\inf_{\bx, \bx' \in \mathsf{X}} \tilde{f}_{\btheta}(\bx,\bx') \text{ and } \tilde{\sigma}_{+} = \sup_{\btheta \in \bTheta}\sup_{\bx, \bx' \in \mathsf{X}} \tilde{f}_{\btheta}(\bx,\bx').$$
\end{lemma}

\begin{proof}
Since $\mathsf{X}$ is compact by (A2) and $f_{\btheta}$ is uniformly bounded under (A3), the mapping defining $\tilde{f}_{\btheta}$ is uniformly continuous under (A1) so that uniform boundedness is preserved:
$$ 0 < \tilde{\sigma}_{-} \leq \tilde{f}_{\btheta}(\bx,\bx') \leq \tilde{\sigma}_{+} < \infty.$$
It follows from Lemma~\ref{l:Step1InhomogenousMarkovChain} that
\begin{align*}
\tilde{p}_{\btheta}(\bx_k | \bx_{k-1}, \bx_{n+1}, \by_{k:n}) \geq \frac{\tilde{\sigma}_{-}}{\tilde{\sigma}_{+}} h_k(\bx_k | \bx_{n+1}, \by_{k:n}),
\end{align*}
where
\begin{align*}
h_k &(\bx_k | \bx_{n+1}, \by_{k:n})\\
&= \frac{
\tilde{g}_{\btheta}(\bx_k,\by_k) 
\int \tilde{f}_{\btheta}(\bx_n,\bx_{n+1}) \prod_{t=k+1}^n \tilde{f}_{\btheta}(\bx_{t-1},\bx_t) \tilde{g}_{\btheta}(\bx_t,\by_t) \diffd \bx_{(k+1):n}}{\int \prod_{t=k}^n \tilde{f}_{\btheta}(\bx_{t-1},\bx_t) \tilde{g}_{\btheta}(\bx_t,\by_t) \diffd \bx_{k:n}}
\end{align*}
is a probability density that does not depend on $\bx_{k-1}$. Note that this result, along with Lemma~\ref{l:Step1InhomogenousMarkovChain}, can be seen as an extension of Lemma~1 of \citet{douc2004}. Hence, by standard results for uniformly minorized Markov chains \citep[see e.g.][Chapter III]{lindvall1992}, it follows that 
$$ \left\Vert \tilde{p}_{\btheta}(\bx_k | \bx_{m-1}, \bx_{n+1}, \by_{m:n}) - 
\tilde{p}_{\btheta}(\bx_k | \bx_{m-1}', \bx_{n+1}, \by_{m:n}) \right\Vert_{\text{TV}} \leq \tau^{k-m+1}. $$
Similarly,
$$ \tilde{p}_{\btheta}(\bx_k | \bx_{m-1}, \bx_{k+1}, \by_{m:k}) \geq
\frac{\tilde{\sigma}_{-}}{\tilde{\sigma}_{+}} \bar{h}_k(\bx_k | \bx_{m-1}, \by_{m:k}),$$
where $\bar{h}_k(\bx_k | \bx_{m-1}, \by_{m:k})$ is another probability density which does not depend on $\bx_{k+1}$. This implies that
$$\left\Vert\tilde{p}_{\btheta}(\bx_k | \bx_{m-1}, \bx_{n+1}, \by_{m:n}) - \tilde{p}_{\btheta}(\bx_k | \bx_{m-1}, \bx'_{n+1}, \by_{m:n})\right\Vert_{\text{TV}} \leq \tau^{n-k+1}.$$
Hence by combining these two results, we obtain for $m \leq k \leq n$
$$\left\Vert\tilde{p}_{\btheta}(\bx_k | \bx_{m-1}, \bx_{n+1}, \by_{m:n}) - \tilde{p}_{\btheta}(\bx_k | \bx_{m-1}', \bx_{n+1}', \by_{m:n})\right\Vert_{\text{TV}} \leq 2\max(\tau^{k-m+1},\tau^{n-k+1}).$$
\end{proof}

This geometric rate of decay for the total variation distance between two conditional densities differing by their initial and end points means in particular that the influence of the initial states $\bX_0$ vanishes as $k \rightarrow \infty$.

We now show that $\tilde{\Expec}_{\btheta}\left[\beta(\bx_t,\by_t,\btheta) | \by_{1:T}\right]$ can be approximated by a stationary process, with similar arguments applying to $\tilde{\Expec}_{\btheta}\left[\alpha(\bx_{t-1},\bx_t,\btheta) | \by_{1:T}\right]$. For $s>0$ and an arbitrary (albeit fixed) value $\bx^\star \in \mathsf{X}$, we define the following stochastic process:
\begin{align*}
\xi_{k,s} &= \xi_{s,\btheta,\bx^\star}(\by_{(k-s):(k+s)})\\
&= \int \beta(\bx_k,\by_k,\btheta)
\tilde{p}_{\btheta}(\bx_k | \bx_{k-s-1}=\bx_{k+s+1}=\bx^\star, \by_{(k-s):(k+s)}) \diffd \bx_k.
\end{align*}
Note that $\xi_{k,s}$ is defined for any $k \in \mathbb{Z}$ and any (doubly-infinite) $\by \in \mathsf{Y}^{\mathbb{Z}}$ because the conditional robustified density $\tilde{p}_{\btheta}(\bx_{n:m} | \bx_{m-1},\bx_{n+1},\by_{n:m})$ is defined for any $n<m$, i.e.\ without the restrictions $0<m$ and $n<T$, by the formula at the beginning of the proof of Lemma~\ref{l:Step1InhomogenousMarkovChain}.

\begin{lemma}\label{l:Step3ApproxStProc}
Under \textnormal{(A1)}--\textnormal{(A3)}, the process $\xi_{k,s}$ is stationary and ergodic if $\bY$ is stationary and ergodic, and 
$$ \sup_{\by_{1:T} \in \mathsf{Y}^T}\left|\sumtT \tilde{\Expec}_{\btheta}\left[\beta(\bx_t,\by_t,\btheta) | \by_{1:T}\right] - \sum_{k=s}^{T-s} \xi_{k,s}\right| \leq 2 (s+T\tau^{s+1}) \sup_{\bx \in \mathsf{X}}\sup_{\by \in \mathsf{Y}}\sup_{\btheta \in \bTheta} |\beta(\bx, \by, \btheta)|.$$
Moreover, for any $u>s$ and any $k \in \mathbb{Z}$
\begin{equation*}
\sup_{\by_{k-u,k+u}} |\xi_{k,s} - \xi_{k,u}| \leq 2 \tau^{s+1} \sup_{\bx \in \mathsf{X}}\sup_{\by \in \mathsf{Y}}\sup_{\btheta \in \bTheta} | \beta(\bx, \by, \btheta)|.
\end{equation*}
\end{lemma}

\begin{proof}
First, for any $s>0$, $\xi_{k,s}$ is stationary and ergodic for any stationary and ergodic process $(\bY_{k})_{k \in \mathbb{Z}}$ because the sequence subsetting and integration in the definition of $\xi_{k,s}$ above preserve such properties. We note that, by Lemma~1 of \citet{leroux1992}, under (A3) $(\bY_{k})_{k \in \mathbb{Z}}$ is stationary and ergodic under the assumed (nominal) SSM represented by $\Plaw_{\btheta}$, for all $\btheta \in \bTheta$. Next, for any $0 < m \leq k \leq n< T$ it holds
\begin{align*}
\tilde{p}_{\btheta}(\bx_k | \by_{1:T}) &=
\int \tilde{p}_{\btheta}(\bx_k | \bx_{m-1}, \bx_{n+1}, \by_{1:T}) \tilde{p}_{\btheta}(\bx_{m-1}, \bx_{n+1} | \by_{1:T}) \diffd \bx_{m-1} \diffd \bx_{n+1}\\
&= \int \tilde{p}_{\btheta}(\bx_k | \bx_{m-1}, \bx_{n+1}, \by_{n:m}) \tilde{p}_{\btheta}(\bx_{m-1}, \bx_{n+1} | \by_{1:T}) \diffd \bx_{m-1} \diffd \bx_{n+1}.
\end{align*}
Lemma~\ref{l:Step2ForgetTVnorm} implies that for any $s>0$, any $s \leq k \leq T-s$ and some fixed $\bx^\star \in \mathsf{X}$
$$ \left\Vert \tilde{p}_{\btheta}(\bx_k | \by_{1:T}) - \tilde{p}_{\btheta}(\bx_k | \bx_{k-s-1} = \bx_{k+s+1} = \bx^\star , \by_{(k-s):(k+s)})\right\Vert_{\text{TV}} \leq 2 \tau^{s+1}.$$
It then follows that
$$ \sup_{\by_{1:T} \in \mathsf{Y}^T} \left\vert \tilde{\Expec}_{\btheta}\left[\beta(\bx_k,\by_k,\btheta) | \by_{1:T}\right] - \xi_{k,s} \right\vert \leq 2 \tau^{s+1} \sup_{\bx \in \mathsf{X}}\sup_{\by \in \mathsf{Y}}\sup_{\btheta \in \bTheta} | \beta(\bx, \by, \btheta)|,$$
which implies the first claim after summation. Note that the supremum of $\beta(\bx, \by, \btheta)$ is guaranteed to be finite under (A1) and (A3). The second claim follows similarly.
\end{proof}

By Lemma~\ref{l:Step3ApproxStProc}, $\xi_k = \lim_{s\rightarrow\infty} \xi_{k,s}$ exists since $(\xi_{k,s})_{s \in \mathbb{N}}$ is a Cauchy sequence and
\begin{equation*}
\sup_{\by_{-\infty:\infty}} |\xi_{k,s} - \xi_k| \leq 2 \tau^{s+1} \sup_{\bx \in \mathsf{X}}\sup_{\by \in \mathsf{Y}}\sup_{\btheta \in \bTheta} | \beta(\bx, \by, \btheta)|.
\end{equation*}
Lemma~\ref{l:Step3ApproxStProc} also shows that $\tilde{\Expec}_{\btheta}\left[\beta(\bx_t,\by_t,\btheta) | \by_{1:T}\right]$ is approximately stationary and ergodic which allows us to invoke Birkhoff's ergodic theorem to obtain
$$ \liminf_{T \rightarrow \infty} \frac{1}{T} \sumtT \tilde{\Expec}_{\btheta}\left[\beta(\bx_t,\by_t,\btheta) | \by_{1:T}\right] \geq \Expec[\xi_{1,s}] - 2 \tau^{s+1} \sup_{\bx \in \mathsf{X}}\sup_{\by \in \mathsf{Y}}\sup_{\btheta \in \bTheta} | \beta(\bx, \by, \btheta)|, \quad \text{a.s.,}$$ 
where the expectation on the right-hand side is with respect to the distribution of any stationary and ergodic $\bY$ process. In particular, it applies to the distribution of $\bY$ under the assumed SSM. Combined with the analog inequality for the limit superior, it follows that the almost sure limit of $\frac{1}{T} \sumtT \tilde{\Expec}_{\btheta}\left[\beta(\bx_t,\by_t,\btheta) | \by_{1:T}\right]$ for $T \rightarrow \infty$ exists and is equal to $\Expec[\xi_1]$. Similar arguments apply to $\alpha(\bx_{t-1},\bx_t,\btheta)$ so that we can conclude that under (A1)--(A4) $T^{-1}\tilde{\bs}(\btheta;\by_{1:T})$ converges to a non-random function as $T \rightarrow \infty$. This non-random function is the expectation $\Expec[\bzeta(\btheta;\bY_0,\bY_{\pm 1},\bY_{\pm 2},\ldots)]$ taken under the invariant distribution of some (doubly-infinite) stationary and ergodic $\bY$ process, where $\bzeta: \bTheta \times \mathsf{Y}^{\mathbb{Z}} \rightarrow \mathbb{R}^p$ is a bounded measurable function.

Finally, we turn to the corrected robustified score. Under (A1)--(A4), the convergence of the correction term $b_T(\btheta)$ follows from the results above and Birkhoff's ergodic theorem:
$$ \lim_{T \rightarrow \infty} T^{-1} b_T(\btheta) = \lim_{T \rightarrow \infty} T^{-1} \Expec_{\btheta} \left[ \tilde{\bs}(\btheta;\by_{1:T}) \right]  =  \Expec_{\btheta}[\bzeta(\btheta;\bY_0,\bY_{\pm 1},\bY_{\pm 2},\ldots)].$$ 
If we set
\begin{equation*}
\bzeta^\star(\btheta;\bY_0,\bY_{\pm 1},\bY_{\pm 2},\ldots) = \bzeta(\btheta;\bY_0,\bY_{\pm 1},\bY_{\pm 2},\ldots) - \lim_{T\rightarrow\infty}T^{-1}b_T(\btheta),
\end{equation*}
the (normalized) corrected robustified score $T^{-1}\tilde{\bs}^\star(\btheta;\by_{1:T}) = T^{-1}\big(\tilde{\bs}(\btheta;\by_{1:T}) - b_T(\btheta)\big)$ thus converges a.s.\ to
$$\Expec[\bzeta^\star(\btheta;\bY_0,\bY_{\pm 1},\bY_{\pm 2},\ldots)] = \int \bzeta^\star(\btheta;\by_0,\by_{\pm 1},\by_{\pm 2},\ldots) \, \Plaw(\text{d} \by)$$
which is bounded for any stationary and ergodic $\bY$. This applies in particular to the marginal distribution of $\bY$ under the assumed SSM so that the corrected robust estimator admits the statistical functional representation $S^\star(\Plaw)$ defined as the solution in $\btheta$ to $\int \bzeta^\star(\btheta;\by_0,\by_{\pm 1},\by_{\pm 2},\ldots) \, \Plaw(\text{d} \by) = \boldsymbol{0}$. The latter satisfies $S^\star(\Plaw_{\btheta}) = \btheta$ for all $\btheta \in \bTheta$ by construction, hence the robust estimator is Fisher consistent.

\section*{Appendix S2: Proof of Theorem~2}

The proof of Theorem~2 in the main body of the paper amounts to verifying the conditions of Theorem~4.2 of \citet{martin1986}. Our setting is slightly different from theirs in that our $\bzeta^\star$ function depends on both past and future observations ($\bY$ is a doubly infinite sequence, see e.g.\ the definition of $\xi_{k,s}$ in the proof of Theorem~1 above). But apart from minor notation changes this does not impact their results and proof.

Conditions (a) and (b) in Theorem~4.2 of \citet[][p.~792]{martin1986} are directly equivalent to our (A5) and (A6), respectively. Conditions (c), (d) and (e) in Theorem~4.2 of \citet{martin1986} follow from the definition of $\bzeta^\star$, see in particular the properties of $\xi_{k,s}$ derived in Lemma~\ref{l:Step2ForgetTVnorm} and \ref{l:Step3ApproxStProc} above. The expression for the IF given in Theorem~2 then follows from Equation~(4.2), (4.2') and (4.6) of \citet{martin1986}.

In the case of patchy outliers with arbitrary but fixed patch length $k$, the existence of $\lim_{\epsilon\downarrow 0} \Expec_{\Plaw_{\btheta}^{\epsilon}}[\bzeta^\star(\btheta;\bY_0^\epsilon,\bY^\epsilon_{\pm 1},\bY^\epsilon_{\pm 2},\ldots)]/\epsilon$ is guaranteed by Theorem~4.2 of \citet{martin1986}. By their Equation~(4.8) and (4.8') we obtain the particular representation of the IF:
$$ \textnormal{IF}(\Plaw_{\bW},S^\star,\Plaw_{\btheta}) = - \frac{1}{k}\bM(\btheta)^{-1} \sum_{j=-\infty}^{\infty} \Expec\left[ \bzeta^\star(\btheta;\bY_{-\infty:(j-k)}^0, \bW_{(j-k+1):j} , \bY_{(j+1):\infty}^0) \right],$$
where the expectation is taken under the joint distribution of $\bY$ and $\bW$ (recall that $\bW$ is assumed stationary and ergodic). Because $\bzeta^\star$ is bounded, it is thus clear that the IF is itself bounded as long as $k$ is fixed.

\section*{Appendix S3: Stationary Version of the North Sea Pollock Assessment Model}

We describe here the modifications made to the dynamics of the unobserved states to create a stationary version of the North Sea pollock stock assessment model used in the simulation study of Section~4 and for the real data example in Section~5. The observation equations for $C_{a,t}$ and $I_{a,t}$ remain unchanged. We stress that these modifications are made solely to define a stationary model and we acknowledge that the resulting dynamics may not be as realistic as the original ones from a marine ecology point of view.

Recall that the vector of unobserved states at year $t$ is
\begin{align*}
\bX_t &= (\log F_{3,t},\ldots,\log F_{9+,t},\log N_{3,t}, \ldots, \log N_{10+,t})^\top,
\end{align*}
with dimension $q=15$, where $F_{a,t}$ is the unitless fishing mortality (instantaneous) rate and $N_{a,t}$ denotes the abundance of fish as counts. The vector of observations at year $t$ is
\begin{align*}
\bY_t &= (\log C_{3,t},\ldots,\log C_{10+,t},\log I_{3,t},\ldots,\log I_{8+,t})^\top,
\end{align*}
with dimension $r=14$, where $C_{a,t}$ represents the total commercial catches (as counts) and $I_{a,t}$ is a unitless index of relative abundance coming from surveys with a standardized fishing effort.

The fishing mortality $F_{a,t}$ now follows an AR(1) on the log scale with age-specific stationary mean:
\begin{align*}
\left[\begin{array}{c} \log F_{3,t}\\ \vdots \\ \log F_{9+,t}  \end{array}\right] &= \left[\begin{array}{c} (1-\phi_F) \lambda_{a=3} + \phi_F \log F_{3,t-1}\\ \vdots \\ (1-\phi_F)\lambda_{a=9+} + \phi_F\log F_{9+,t-1}  \end{array}\right] + \bxi_t, \quad \forall t=1,\ldots,T,
\end{align*}
where $\lambda_a$ is the stationary invariant expectation of $\log F_{a,t}$ as $t\rightarrow\infty$ for $a=3,\ldots,9+$, the constraint $|\phi_F|<1$ is enforced in the estimation, and the $\bxi_t$s have the same multivariate normal distribution with covariance $\mathbf{\Sigma}$. The initial state distribution at $t=1$ corresponds to the stationary multivariate Gaussian with expectation vector $\blambda = (\lambda_{a=3}, \ldots, \lambda_{a=9+})^\top$ and covariance $\mathbf{\Sigma}/(1-\phi_F^2)$. The stationary means $\blambda$ and AR coefficient $\phi_F$ are additional parameters to estimate.

Regarding the true abundance $N_{a,t}$, for the youngest age class $a=3$ we specify an AR(1) on the log scale with stationary mean $\gamma$:
\begin{align*}
\log N_{3,t} &= (1-\phi_R)\gamma + \phi_R \log N_{3,t-1} + \eta_{3,t},
\end{align*}
where the $\eta_{3,t}$s again are i.i.d.\ N$(0,\sigma^2_R)$. The initial state is given by the stationary Gaussian distribution with mean $\gamma$ and variance $\sigma^2_R/(1-\phi_R^2)$, where $\gamma$ and $\phi_R$ are additional parameters to estimate.

For $4\leq a \leq 9$, we have the modified survival equation
\begin{align}\label{eq:survival49}
\log N_{a,t} &= \phi_N \big(\log N_{a-1,t-1} - F_{a-1,t-1} - M_{a-1,t-1}\big) + \eta_{a,t},
\end{align}
where $\phi_N \in (-1,+1)$ tunes the magnitude of the stationary means of $\log N_{a,t}$ across age classes based on the value of $\gamma$, and $\eta_{a,t}$ are again i.i.d.\ N$(0,\sigma^2_N)$. The stationary distribution of this process does not admit a tractable density and moments because of the stochasticity in $F_{a-1,t-1}$, resulting in a linear combination of independent Gaussian and non-independent log-normally distributed random variables. We thus assume a Gaussian distribution with matching moments as an approximation to the stationary distribution for the initial states, see the next section for details.

Finally, for the largest age class $a=10+$ we add the abundance numbers on the log scale so as to preserve some linearity and consider autoregressive dynamics:
\begin{align}\label{eq:survival10}
\log N_{10+,t} &= \phi_N \big(\log N_{9,t-1} - F_{9+,t-1} - M_{9,t-1}\big)\nonumber\\
&\phantom{\mathrel{=}}\, + \phi_P \big(\log N_{10+,t-1} - F_{9+,t-1} - M_{10+,t-1}\big)  + \eta_{10+,t},
\end{align}
where $\phi_P$ is constrained within $(-1,+1)$ and the $\eta_{10+,t}$s are i.i.d.\ N$(0,\sigma^2_P)$. The stationary distribution is again not tractable and we consider an approximate Gaussian distribution for the initial states, see the next section.

With these modified states dynamics, the parameter vector to estimate now is
\begin{align*}
\btheta = (&\lambda_{a=3}, \ldots, \lambda_{a=9+}, \phi_F, \sigma_{F_{a=3}}, \sigma_{F_{a\geq4}}, \rho, \gamma,\\
&\phi_R, \sigma_R, \phi_N, \sigma_N, \phi_P, \sigma_P, \sigma_C, q_{a=3}, \ldots, q_{a=8+}, \sigma_I)^\top
\end{align*}
with $p=26$.

\subsection*{Derivation and approximation of stationary distributions}

The modified survival equations (\ref{eq:survival49}) and (\ref{eq:survival10}) above lead to stationary dynamics yet the invariant distributions and some moments seem intractable. We give here some details about the derivation and the need for some approximations.

For the intermediate age classes $4 \leq a \leq 9$, the autoregressive recursions and limit as $t\rightarrow\infty$ of $\log N_{3,t}$ lead to
\begin{align}
\log N_{a,t} &= \phi_N^{a-3}\phi_R^{t-a+3}\log N_{3,0} + \phi_N^{a-3}\gamma + \phi_N^{a-3}\sum_{j=0}^{t-1}\phi_R^j\eta_{3,t-j-a+3}\nonumber \\
&\phantom{\mathrel{=}}\; + \sum_{j=0}^{a-4}\phi_N^j \eta_{a-j,t-j} - \sum_{j=1}^{a-3} \phi_N^j (F_{a-j,t-j} + M_{a-j,t-j}),
\end{align}
with the limit being
\begin{align}\label{eq:logNat}
\lim_{t\rightarrow\infty}\log N_{a,t} &= \phi_N^{a-3}\gamma + \lim_{t\rightarrow\infty} \phi_N^{a-3}\sum_{j=0}^{t-1}\phi_R^j\eta_{3,t-j-a+3}\nonumber \\
&\phantom{\mathrel{=}}\; + \lim_{t\rightarrow\infty}\sum_{j=0}^{a-4}\phi_N^j \eta_{a-j,t-j} - \lim_{t\rightarrow\infty}\sum_{j=1}^{a-3} \phi_N^j (F_{a-j,t-j} + M_{a-j,t-j}).
\end{align}
Making use of the fact that the natural mortality is assumed here a fixed covariate constant through time ($M_{a,t} = M_a$ $\forall t$), the expectation taken over the joint distribution including $\bF=(F_3,\ldots,F_9)^\top$ is thus
$$ \Expec[\log N_{a}] = \phi_N^{a-3}\gamma - \sum_{j=1}^{a-3} \phi_N^j (\Expec[F_{a-j}] + M_{a-j}) $$
and the variance is
\begin{align*}
\Var[\log N_{a}] &= \phi_N^{2(a-3)}\frac{\sigma^2_R}{1-\phi_R^2} + \sigma^2_N\frac{1-\phi_N^{2(a-3)}}{1-\phi_N^2}\\
&\phantom{\mathrel{=}}\; + \lim_{t\rightarrow\infty} \sum_{j=1}^{a-3}\sum_{i=1}^{a-3}\phi_N^j\phi_N^i\Cov[F_{a-j,t-j},F_{a-i,t-i}]\\
&= \phi_N^{2(a-3)}\frac{\sigma^2_R}{1-\phi_R^2} + \sigma^2_N\frac{1-\phi_N^{2(a-3)}}{1-\phi_N^2}\\
&\phantom{\mathrel{=}}\; + \sum_{j=1}^{a-3}\sum_{i=1}^{a-3}\phi_N^j\phi_N^i \exp\left(\lambda_{a-j}+\lambda_{a-i}+\frac{\sigma_{F_{a-j}}^2+\sigma_{F_{a-i}}^2}{2(1-\phi_F^2)}\right)\\
&\phantom{\mathrel{=}\sum_{j=1}^{a-3}\sum_{i=1}^{a-3}}\quad\times\left(\exp\left(\phi_F^{|i-j|}\rho^{|i-j|}\frac{\sigma_{F_{a-j}}\sigma_{F_{a-i}}}{1-\phi_F^2}\right) - 1\right),
\end{align*}
where the last term is derived from
\begin{align*}
\lim_{t\rightarrow\infty}\Cov[\log F_{a-i,t-i},\log F_{a-j,t-j}] &= \phi_F^{|i-j|}\Cov[\log F_{a-j}, \log F_{a-i}]\\
&= \phi_F^{|i-j|} \mathbf{\Sigma}_{a-i,a-j}/(1-\phi_F^2).
\end{align*}

In Equation~(\ref{eq:logNat}), a linear combination of non-independent log-normal and independent Gaussian random variables does not admit a distribution with an analytical density. Therefore the stationary distribution of $N_{a,t}$, and in particular its expectation, cannot easily be derived. A simple way to approximate it is by assuming $\log N_{a,t}$ is Gaussian, in effect neglecting the stochasticity in $\bF$, so that $N_{a,t}$ is approximated by a log-normal with matching mean and variance. We thus obtain the approximated expectation
\begin{align*}
\Expec[N_{a}] &\approxeq \exp\big(\Expec[\log N_a]+\Var[\log N_a]/2\big).
\end{align*}
Simulations (not presented here) have shown that the bias due to this approximation is negligible given the simulation design in the main body of the paper, with samples of varying size up to $T=10\,000$ generated under the stationary model.

For the oldest age class $a=10$ (dropping the $10+$ notation here for legibility), the fact that $a=9$ is a plus group for the fishing mortalities and thus that $F_{a=9} = F_{a=10}$, the autoregressive recursions lead to
\begin{align*}
\log N_{10,t} &= \phi_P^t\log N_{10,0} + \phi_N \sum_{j=0}^{t-1} \phi_P^j \left( \log N_{10-1,t-j-1} - F_{10-1,t-j-1} - M_{10-1,t-j-1}\right)\nonumber \\
&\phantom{\mathrel{=}}\; - \sum_{j=1}^t \phi_P^j (F_{10,t-j} + M_{10,t-j}) + \sum_{j=0}^{t-1} \phi_P^j \eta_{10,t-j}.
\end{align*}
Taking the limit yields
\begin{align}\label{eq:logNAt}
\lim_{t\rightarrow\infty}\log N_{10,t} &= \lim_{t\rightarrow\infty} \sum_{j=0}^{t-1} \phi_P^j \eta_{10,t-j} - \lim_{t\rightarrow\infty} \sum_{j=1}^t \phi_P^j (F_{10,t-j} + M_{10,t-j})\nonumber\\
&\phantom{\mathrel{=}}\; + \lim_{t\rightarrow\infty} \phi_N \sum_{j=0}^{t-1} \phi_P^j \left( \log N_{10-1,t-j-1} - F_{10-1,t-j-1} - M_{10-1,t-j-1}\right)\nonumber\\
&= \lim_{t\rightarrow\infty} \sum_{j=0}^{t-1} \phi_P^j \eta_{10,t-j} - \lim_{t\rightarrow\infty} \sum_{j=1}^t \phi_P^j (F_{10,t-j} + M_{10,t-j})\nonumber\\
&\phantom{\mathrel{=}}\; + \lim_{t\rightarrow\infty} \sum_{j=0}^{t-1} \phi_P^j \bigg[ \phi_N^{10-3}\gamma + \phi_N^{10-3}\sum_{k=0}^{t-j-2}\phi_R^k\eta_{3,t-j-k-10+3}\nonumber\\
&\phantom{\mathrel{=}\; \lim_{t\rightarrow\infty} \sum_{j=0}^{t-1} \phi_P^j\bigg[\;}\;\, + \sum_{k=1}^{10-4}\phi_N^k \eta_{10-k,t-j-k} \nonumber\\
&\phantom{\mathrel{=}\; \lim_{t\rightarrow\infty} \sum_{j=0}^{t-1} \phi_P^j\bigg[\;}\;\, - \sum_{k=1}^{10-3} \phi_N^k (F_{10-k,t-j-k} + M_{10-k,t-j-k}) \bigg].
\end{align}
The stationary expectation is thus
\begin{align*}
\Expec[\log N_{10}] &= \frac{1}{1-\phi_P}\left(\phi_N^{10-3}\gamma -\phi_P(\Expec[F_{10}] + M_{10}) - \sum_{j=1}^{10-3} \phi_N^j (\Expec[F_{10-j}] + M_{10-j}) \right)
\end{align*}
and the variance is
\begin{align}\label{eq:VarlogNAt}
\Var[\log N_{10}] &= \frac{\sigma^2_P}{1-\phi_P^2} + \sigma_N^2\frac{\phi_N^2}{1-\phi_P^2}\frac{1-\phi_N^{2(10-4)}}{1-\phi_N^2}\nonumber\\
&\phantom{\mathrel{=}}\; + \sigma^2_R\frac{\phi_N^{2(10-3)}}{(\phi_R-\phi_P)^2}\left(\frac{\phi_R^2}{1-\phi_R^2} + \frac{\phi_P^2}{1-\phi_P^2} -2\frac{\phi_R\phi_P}{1-\phi_R\phi_P}\right)\nonumber\\
&\phantom{\mathrel{=}}\; + \lim_{t\rightarrow\infty} \Var\left[\sum_{j=1}^t \phi_P^j F_{10,t-j} + \sum_{j=0}^{t-1}\phi_P^j \sum_{k=1}^{10-3} \phi_N^k F_{10-k,t-j-k} \right]\nonumber\\
&= \frac{\sigma^2_P}{1-\phi_P^2} + \sigma_N^2\frac{\phi_N^2}{1-\phi_P^2}\frac{1-\phi_N^{2(10-4)}}{1-\phi_N^2}\nonumber\\
&\phantom{\mathrel{=}}\; + \sigma^2_R\frac{\phi_N^{2(10-3)}}{(\phi_R-\phi_P)^2}\left(\frac{\phi_R^2}{1-\phi_R^2} + \frac{\phi_P^2}{1-\phi_P^2} -2\frac{\phi_R\phi_P}{1-\phi_R\phi_P}\right)\nonumber\\
&\phantom{\mathrel{=}}\; + \lim_{t\rightarrow\infty}\sum_{j=1}^t\sum_{i=1}^t\phi_P^{i+j}\Cov[F_{10,t-j},F_{10,t-i}]\nonumber\\
&\phantom{\mathrel{=}}\; + \lim_{t\rightarrow\infty}\sum_{j=0}^{t-1}\sum_{i=0}^{t-1}\phi_P^{i+j} \sum_{k=1}^{10-3}\sum_{l=1}^{10-3}\phi_N^{k+l} \Cov[F_{10-k,t-j-k},F_{10-l,t-i-l}]\nonumber\\
&\phantom{\mathrel{=}}\; + 2\lim_{t\rightarrow\infty}\sum_{j=1}^{t}\sum_{i=0}^{t-1}\phi_P^{i+j} \sum_{k=1}^{10-3}\phi_N^k \Cov[F_{10,t-j},F_{10-k,t-i-k}],
\end{align}
where
\begin{align*}
\lim_{t\rightarrow\infty}\Cov[F_{a,t},F_{b,t-s}] &= \exp\left(\lambda_{a} + \lambda_{b} + \frac{\sigma^2_{F_{a}} + \sigma^2_{F_{b}}}{2(1-\phi_F^2)}\right)\left(\exp\left(\phi_F^{|s|}\rho^{|a-b|}\frac{\sigma_{F_{a}}\sigma_{F_{b}}}{1-\phi_F^2}\right)-1\right)
\end{align*}
based on the fact that both $\log F_{a,t}$ and $\log F_{b,t-s}$ are Gaussian and that
$$\lim_{t\rightarrow\infty} \Cov[\log F_{a,t},\log F_{b,t-s}] = \phi_F^{|s|}\mathbf{\Sigma}_{a,b}/(1-\phi_F^2),$$
for $3 \leq a,b \leq 10$ and for any integer $s$.

The last three terms of $\Var[\log N_{10}]$ in (\ref{eq:VarlogNAt}) involve infinite convergent sums for which no general closed form is readily available due to the indices appearing in the exponential function. Since the sums converge fairly quickly for all realistic values of $\phi_P$ and $\phi_N$, they can be approximated by truncating at some given $t$. Auxiliary simulations (not shown here) confirmed that truncating such sums at $t=20$ ensured negligible error given the simulation design in the main body of the paper.

Again, in (\ref{eq:logNAt}) $\log N_{10,t}$ appears to be a sum of non-independent log-normal and independent Gaussian random variables, the distribution for which no closed form exists to the best of our knowledge. We thus use the same approximation of considering $\log N_{10,t}$ as Gaussian so that the stationary distribution of $N_{10,t}$ can be approximated by a log-normal distribution with corresponding mean and variance. Hence, we obtain
\begin{align*}
\Expec[N_{10}] &\approxeq \exp\big(\Expec[\log N_{10}]+\Var[\log N_{10}]/2\big).
\end{align*}

\FloatBarrier
\newpage 
\section*{Appendix S4: Simulation Study Tables and Figures}

\begin{table}[H]
\caption{Simulation true parameter $\btheta_0$ and starting values for the MLE, from both non-stationary and stationary versions of the model. \label{tab:simparam}}
\begin{center}
\begin{tabular}{rrrcrr}
\hline & \multicolumn{2}{c}{Non-stationary model} & \hphantom{A} & \multicolumn{2}{c}{Stationary model}\\
 & $\btheta_0$ & Starting value && $\btheta_0$ & Starting value\\
\hline
  $\lambda_{a=3}$       &         &      && -1.80    & -1.00 \\ 
  $\lambda_{a=4}$       &         &      && -1.10    & -1.00 \\ 
  $\lambda_{a=5}$       &         &      && -0.92    & -1.00 \\ 
  $\lambda_{a=6}$       &         &      && -0.86    & -1.00 \\ 
  $\lambda_{a=7}$       &         &      && -0.78    & -1.00 \\ 
  $\lambda_{a=8}$       &         &      && -0.63    & -1.00 \\ 
  $\lambda_{a=9+}$      &         &      && -0.28    & -1.00 \\ 
  $\phi_F$              &         &      &&  0.71    & 0.44$^\ast$ \\ 
  $\sigma_{F_{a=3}}$    & 0.20    & 1.00 &&  0.44    & 1.00 \\ 
  $\sigma_{F_{a\geq4}}$ & 0.16    & 1.00 &&  0.27    & 1.00 \\ 
  $\rho$                & 0.88    & 0.46$^\dag$ &&  0.77    & 0.46$^\dag$ \\ 
  $\gamma$              &         &      && 12.50    & 10.00 \\ 
  $\phi_R$              &         &      &&  0.39    & 0.44$^\ast$ \\ 
  $\sigma_R$            & 0.45    & 1.00 &&  0.44    & 1.00 \\ 
  $\phi_N$              &         &      &&  0.98    & 0.46$^\dag$ \\ 
  $\sigma_N$            & 0.17    & 1.00 &&  0.16    & 1.00 \\ 
  $\phi_P$              &         &      &&  0.20    & 0.00$^\ddag$ \\ 
  $\sigma_P$            & 0.22    & 1.00 &&  0.24    & 1.00 \\ 
  $\sigma_C$            & 0.22    & 1.00 &&  0.14    & 1.00 \\ 
  $q_{a=3}$             & 5.1e-05 & 6.7e-03$^\mathsection$ &&  2.6e-05 & 6.7e-03$^\mathsection$ \\
  $q_{a=4}$             & 8.4e-05 & 6.7e-03$^\mathsection$ &&  5.1e-05 & 6.7e-03$^\mathsection$ \\
  $q_{a=5}$             & 6.3e-05 & 6.7e-03$^\mathsection$ &&  4.4e-05 & 6.7e-03$^\mathsection$ \\
  $q_{a=6}$             & 4.2e-05 & 6.7e-03$^\mathsection$ &&  3.5e-05 & 6.7e-03$^\mathsection$ \\
  $q_{a=7}$             & 2.8e-05 & 6.7e-03$^\mathsection$ &&  3.0e-05 & 6.7e-03$^\mathsection$ \\
  $q_{a=8+}$            & 3.0e-05 & 6.7e-03$^\mathsection$ &&  4.1e-05 & 6.7e-03$^\mathsection$ \\
  $\sigma_I$            & 0.63    & 1.00 &&  0.61    & 1.00 \\ 
\hline
\multicolumn{6}{l}{$^\ast${\small 1.00 on the scale of the transformation $\log[(0.95 + \theta)/(0.95 - \theta)]$}}\\
\multicolumn{6}{l}{$^\dag${\small 1.00 on the scale of the transformation $\log[(1 + \theta)/(1 - \theta)]$}}\\
\multicolumn{6}{l}{$^\ddag${\small 0.00 on the scale of the transformation $\log[(0.95 + \theta)/(0.95 - \theta)]$}}\\
\multicolumn{6}{l}{$^\mathsection${\small -5.00 on the scale of the transformation $\log(\theta)$}}\\
\end{tabular}
\end{center}
\end{table}

\begin{figure}[H]
\begin{center}
\includegraphics[width=0.9\textwidth]{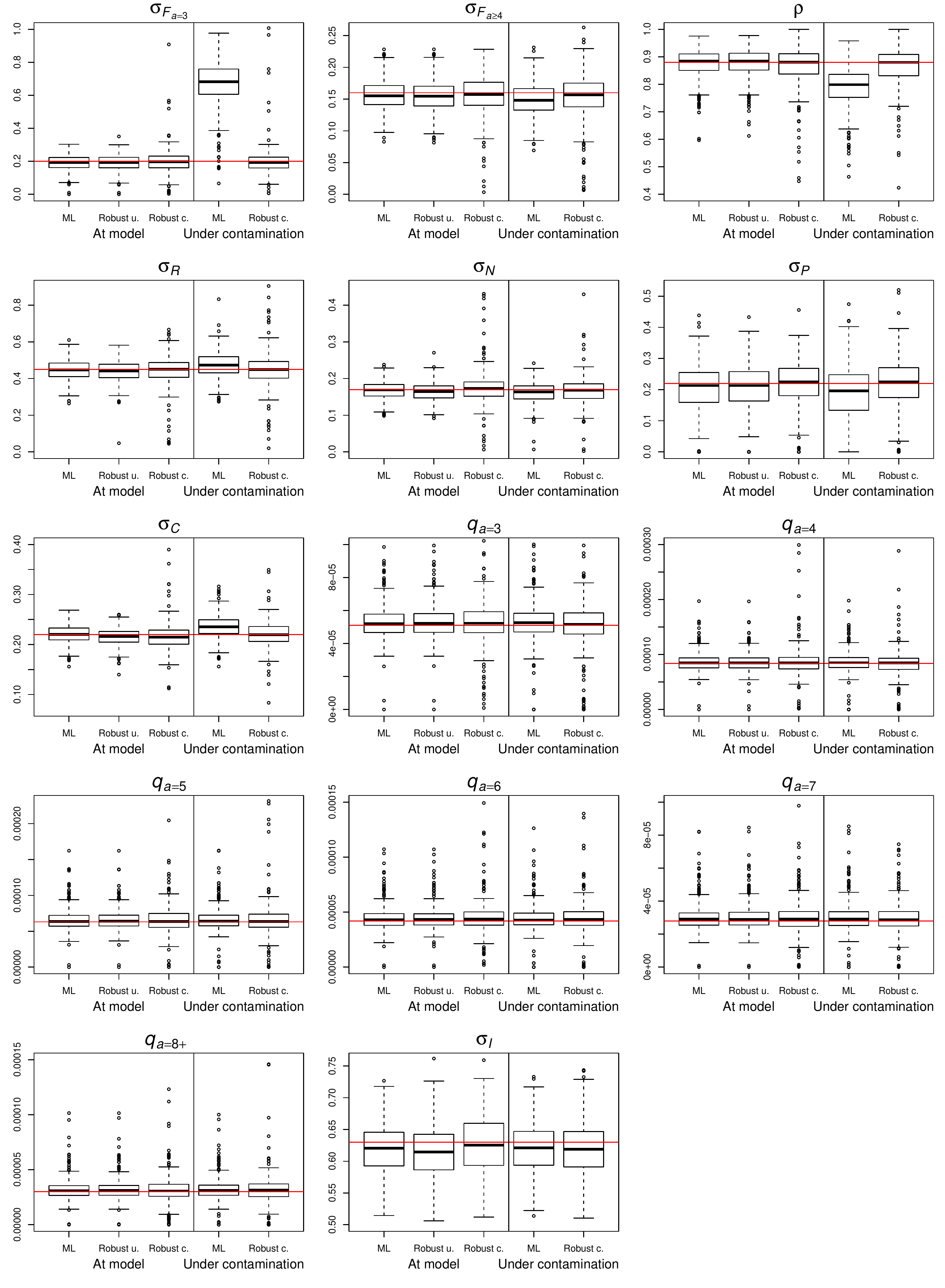}
\end{center}
\caption{Boxplots of ML and robust estimates for all $p=14$ parameters of the non-stationary version of the model, both at the assumed model and under contamination in each panel. ``Robust u.'' stands for the uncorrected robust estimator $\tilde{\btheta}_T^{[1]}$ from Step~1 while ``Robust c.'' is the corrected one from Step~3. The red horizontal solid line is the true parameter value. The vertical scale is manually set to improve visualization, some points (both ML and Robust) are not shown. \label{fig:boxplots_nst}}
\end{figure}

\begin{figure}[H]
\begin{center}
\includegraphics[width=0.9\textwidth]{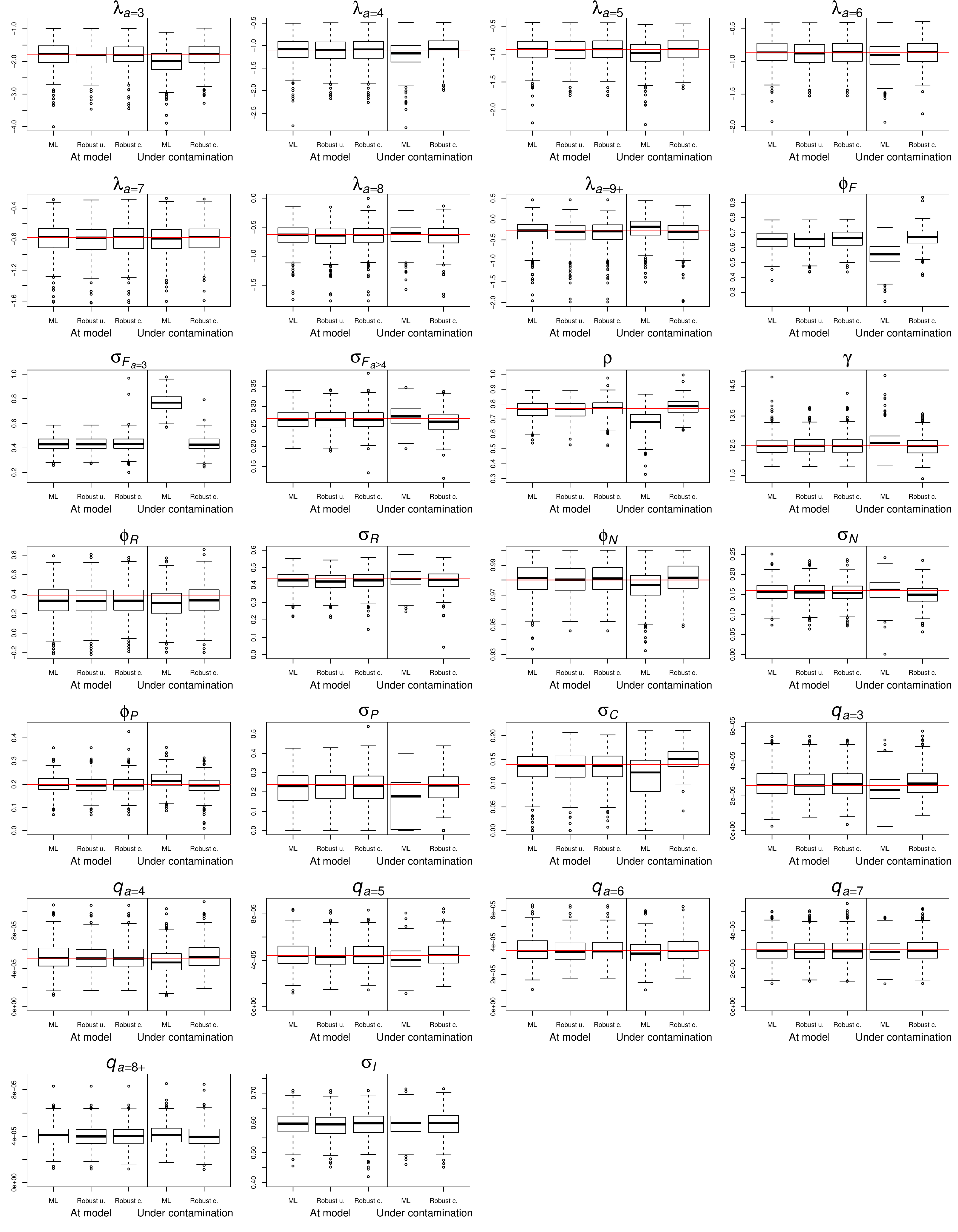}
\end{center}
\caption{Boxplots of ML and robust estimates for all $p=26$ parameters of the stationary version of the model, both at the assumed model and under contamination in each panel. ``Robust u.'' stands for the uncorrected robust estimator $\tilde{\btheta}_T^{[1]}$ from Step~1 while ``Robust c.'' is the corrected one from Step~3. The red horizontal solid line is the true parameter value. The vertical scale is manually set to improve visualization, some points (both ML and Robust) are not shown. \label{fig:boxplots_st}}
\end{figure}

\begin{figure}
\begin{center}
\includegraphics[width=\textwidth]{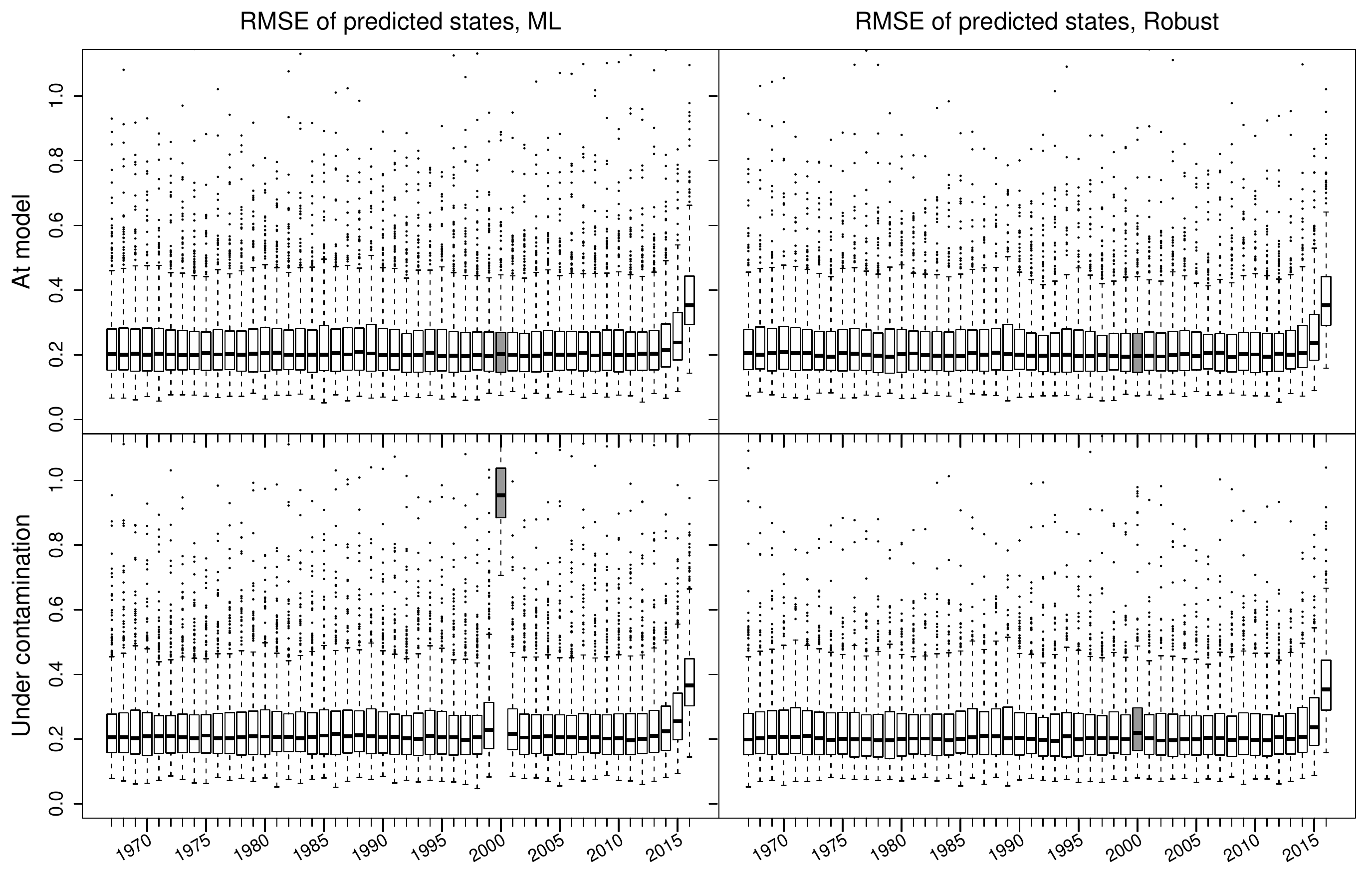} 
\end{center}
\caption{Boxplots of RMSE of predicted states $\bX_t$ for the stationary version of the model. The box corresponding to the contaminated year 2000 is colored in grey. The vertical scale is manually set to improve visualization, some points (both ML and Robust) are not shown. \label{fig:boxplotsRMSEst}}
\end{figure}

\begin{table}[H]
\caption{Quartiles of the robustness weights corresponding to the likelihood contributions of $\log C_{a=3,t}$, for three consecutive years, at the model and under contamination, and for both non-stationary and stationary versions of the model. \label{tab:simrobweights}}
\begin{center}
\begin{tabular}{lrp{1cm}p{1cm}p{1cm}cp{1cm}p{1cm}p{1cm}} 
\hline
& & \multicolumn{3}{c}{At the model} & \hphantom{~} & \multicolumn{3}{c}{Under contamination}\\
 && 1999 & 2000 & 2001 & & 1999 & 2000 & 2001\\
\hline
                &  Q1 & 1.00 & 1.00 & 1.00 && 1.00 & 0.01 & 1.00 \\ 
Non-stationary  &  Q2 & 1.00 & 1.00 & 1.00 && 1.00 & 0.01 & 1.00 \\ 
                &  Q3 & 1.00 & 1.00 & 1.00 && 1.00 & 0.01 & 1.00 \\[1em]
            &  Q1 & 1.00 & 1.00 & 1.00 && 1.00 & 0.00 & 1.00 \\ 
Stationary  &  Q2 & 1.00 & 1.00 & 1.00 && 1.00 & 0.00 & 1.00 \\ 
            &  Q3 & 1.00 & 1.00 & 1.00 && 1.00 & 0.00 & 1.00 \\ 
\hline
\end{tabular}
\end{center}
\end{table}

\FloatBarrier
\newpage 
\section*{Appendix S5: North Sea Pollock Assessment Tables and Figures}

\begin{figure}[H]
\begin{center}
\includegraphics[width=0.9\textwidth]{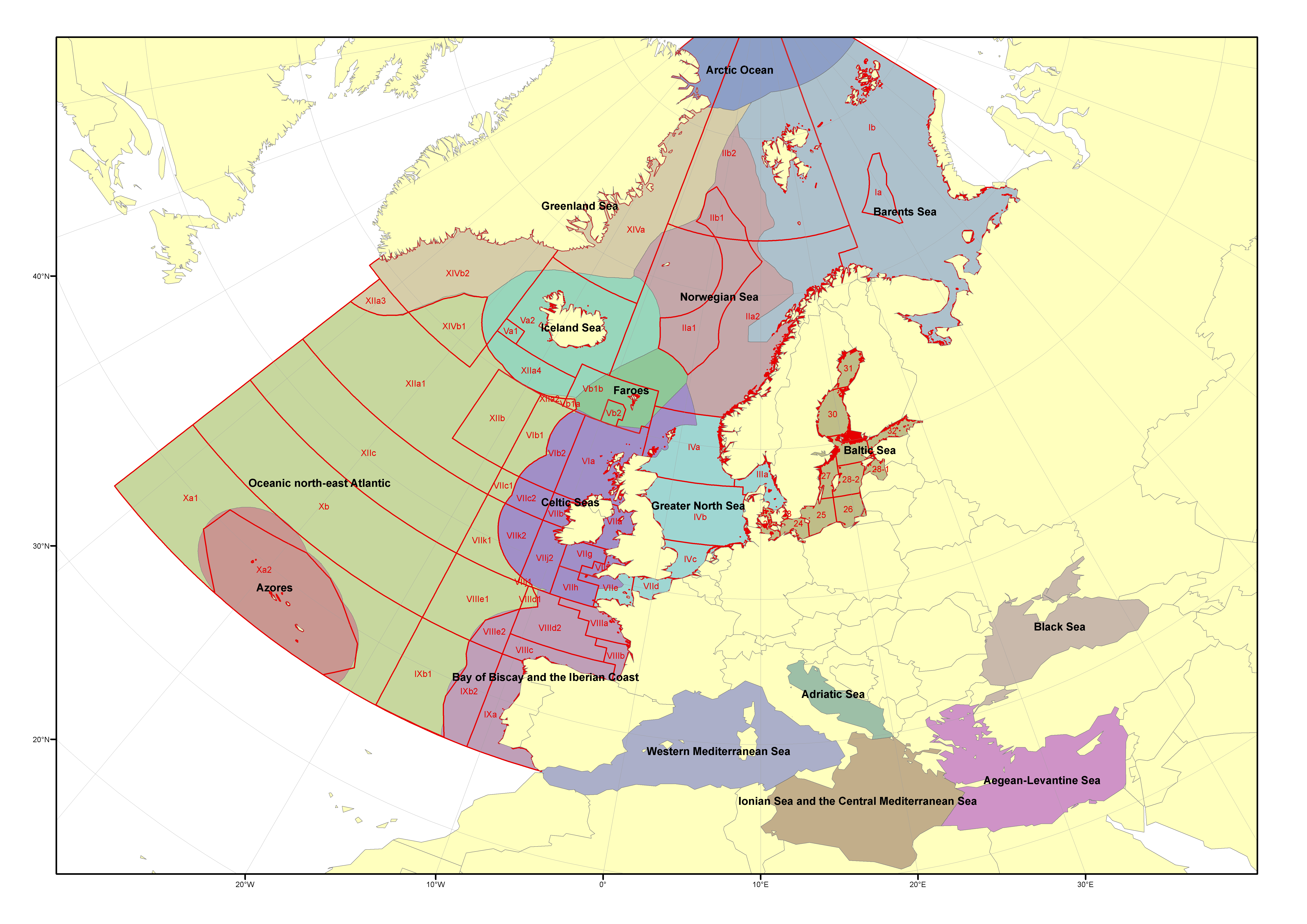}
\end{center}
\caption{Ecoregions and areas as defined by ICES. The North Sea pollock stock spans Division IIIa and Subareas IV and VI. \label{fig:ICESareas}}
\end{figure}

\begin{table}[H]
\caption{Computation times (in seconds) when fitting both versions of the model to the North Sea pollock data with a laptop housing a 2.9 GHz CPU. ``Robust u.'' stands for the uncorrected robust estimator $\tilde{\btheta}_T^{[1]}$ from Step~1 while ``Robust c.'' is the corrected one from Step~3.}
\begin{center}
\begin{tabular}{lrrr}
	\hline
	& ML & Robust u.\ & Robust c.\ \\ 
	\hline
	Non-stationary & $5.9 \times 10^{-1}$ & $7.2 \times 10^{-1}$ & $2.1 \times 10^{3}$ \\ 
	Stationary     & $5.9 \times 10^{0}$  & $4.1 \times 10^{0}$  & $1.1 \times 10^{4}$ \\ 
	\hline
\end{tabular}
\end{center}
\end{table}

\begin{table}[H]
\caption{North Sea pollock data ML and robust estimates of $\btheta$ (standard errors between parentheses) for both non-stationary and stationary versions of the model. \label{tab:dataparam}}
\begin{center}
\begin{tabular}{rrrcrr}
\hline
& \multicolumn{2}{c}{Non-stationary model} & \hphantom{A} & \multicolumn{2}{c}{Stationary model}\\
& ML & Robust && ML & Robust\\
\hline
$\lambda_{a=3}$         &                       &                       && -1.745 (0.347)        & -1.774 (0.868) \\ 
$\lambda_{a=4}$         &                       &                       && -1.090 (0.249)        & -1.111 (0.660) \\ 
$\lambda_{a=5}$         &                       &                       && -0.915 (0.207)        & -0.919 (0.496) \\ 
$\lambda_{a=6}$         &                       &                       && -0.844 (0.176)        & -0.861 (0.365) \\ 
$\lambda_{a=7}$         &                       &                       && -0.763 (0.162)        & -0.776 (0.248) \\ 
$\lambda_{a=8}$         &                       &                       && -0.618 (0.180)        & -0.625 (0.194) \\ 
$\lambda_{a=9+}$        &                       &                       && -0.270 (0.264)        & -0.276 (0.318) \\ 
$\phi_F$                &                       &                       && 0.721 (0.115)         & 0.714 (0.186) \\ 
$\sigma_{F_{a=3}}$      & 0.215 (0.037)         & 0.202 (0.040)         && 0.433 (0.091)         & 0.436 (0.110) \\ 
$\sigma_{F_{a\geq4}}$   & 0.165 (0.028)         & 0.159 (0.027)         && 0.268 (0.043)         & 0.270 (0.051) \\ 
$\rho$                  & 0.885 (0.052)         & 0.882 (0.053)         && 0.774 (0.059)         & 0.772 (0.073) \\ 
$\gamma$                &                       &                       && 12.455 (0.292)        & 12.465 (0.807) \\ 
$\phi_R$                &                       &                       && 0.387 (0.161)         & 0.390 (0.164) \\ 
$\sigma_R$              & 0.455 (0.062)         & 0.448 (0.076)         && 0.427 (0.053)         & 0.436 (0.055) \\ 
$\phi_N$                &                       &                       && 0.975 (0.010)         & 0.975 (0.028) \\ 
$\sigma_N$              & 0.158 (0.028)         & 0.170 (0.028)         && 0.162 (0.029)         & 0.162 (0.039) \\ 
$\phi_P$                &                       &                       && 0.200 (0.035)         & 0.201 (0.053) \\ 
$\sigma_P$              & 0.201 (0.068)         & 0.219 (0.065)         && 0.239 (0.101)         & 0.237 (0.120) \\ 
$\sigma_C$              & 0.226 (0.021)         & 0.216 (0.021)         && 0.147 (0.049)         & 0.141 (0.065) \\ 
$q_{a=3}$               & 5.14e-05 (7.16e-06)   & 5.14e-05 (7.37e-06)   && 2.62e-05 (9.57e-06)   & 2.58e-05 (2.58e-05) \\ 
$q_{a=4}$               & 8.40e-05 (1.16e-05)   & 8.42e-05 (1.17e-05)   && 5.16e-05 (1.51e-05)   & 5.12e-05 (3.96e-05) \\ 
$q_{a=5}$               & 6.29e-05 (8.82e-06)   & 6.29e-05 (8.90e-06)   && 4.47e-05 (1.10e-05)   & 4.45e-05 (2.72e-05) \\ 
$q_{a=6}$               & 4.20e-05 (6.13e-06)   & 4.22e-05 (6.27e-06)   && 3.52e-05 (7.37e-06)   & 3.52e-05 (1.64e-05) \\ 
$q_{a=7}$               & 2.81e-05 (4.38e-06)   & 2.84e-05 (4.58e-06)   && 2.89e-05 (5.40e-06)   & 2.95e-05 (9.89e-06) \\ 
$q_{a=8+}$              & 3.02e-05 (5.19e-06)   & 3.03e-05 (5.32e-06)   && 4.09e-05 (7.87e-06)   & 4.09e-05 (9.48e-06) \\ 
$\sigma_I$              & 0.623 (0.040)         & 0.623 (0.041)         && 0.616 (0.041)         & 0.615 (0.044) \\ 
\hline
\end{tabular}
\end{center}
\end{table}

\begin{figure}[H]
\begin{center}
\includegraphics[width=\textwidth]{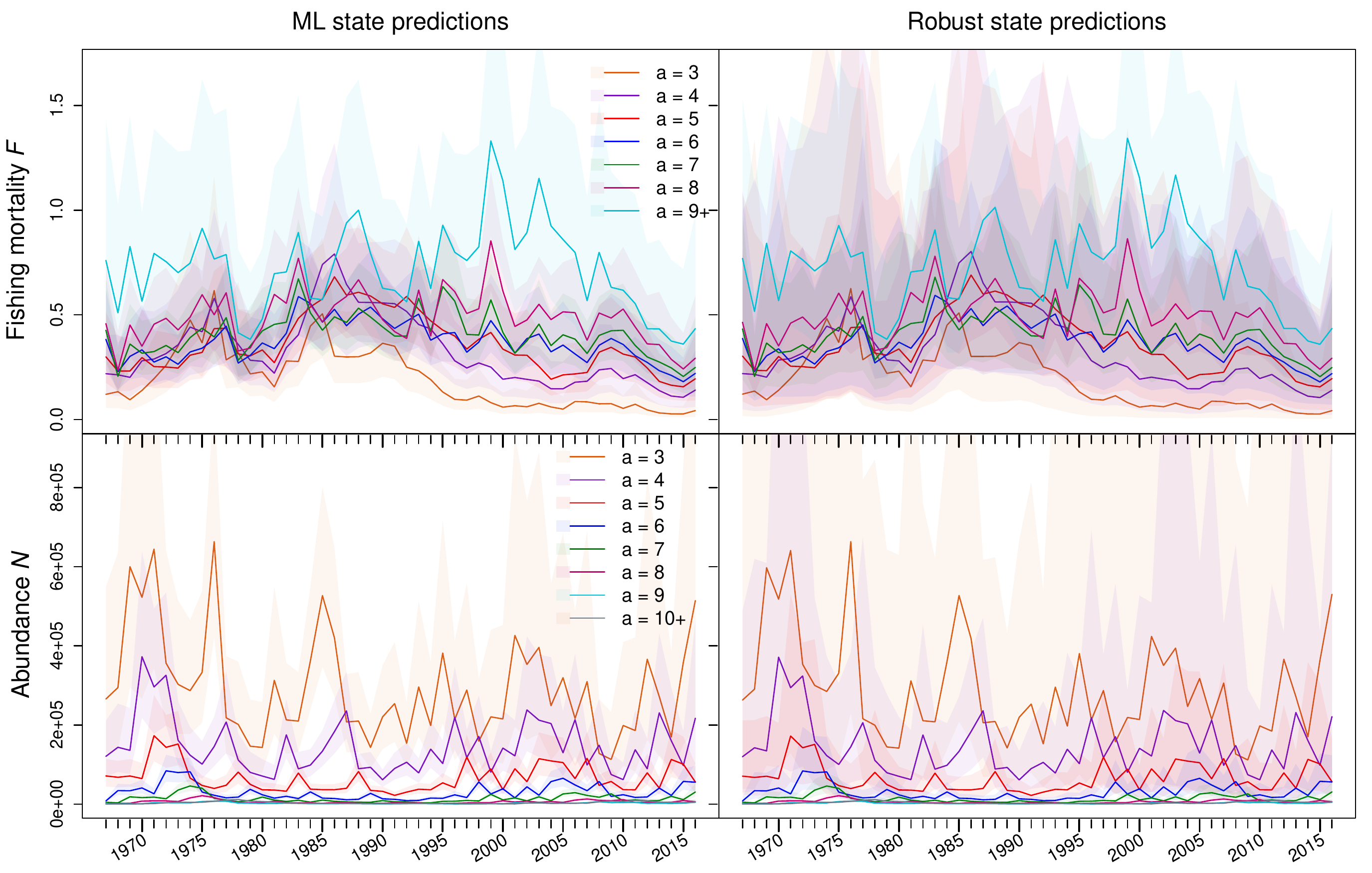} 
\end{center}
\caption{North Sea pollock data ML and robust predictions of $\bX$ (on the original scale) for the stationary version of the model. \label{fig:datapredst}}
\end{figure}

\begin{figure}[H]
\begin{center}
\includegraphics[width=0.9\textwidth]{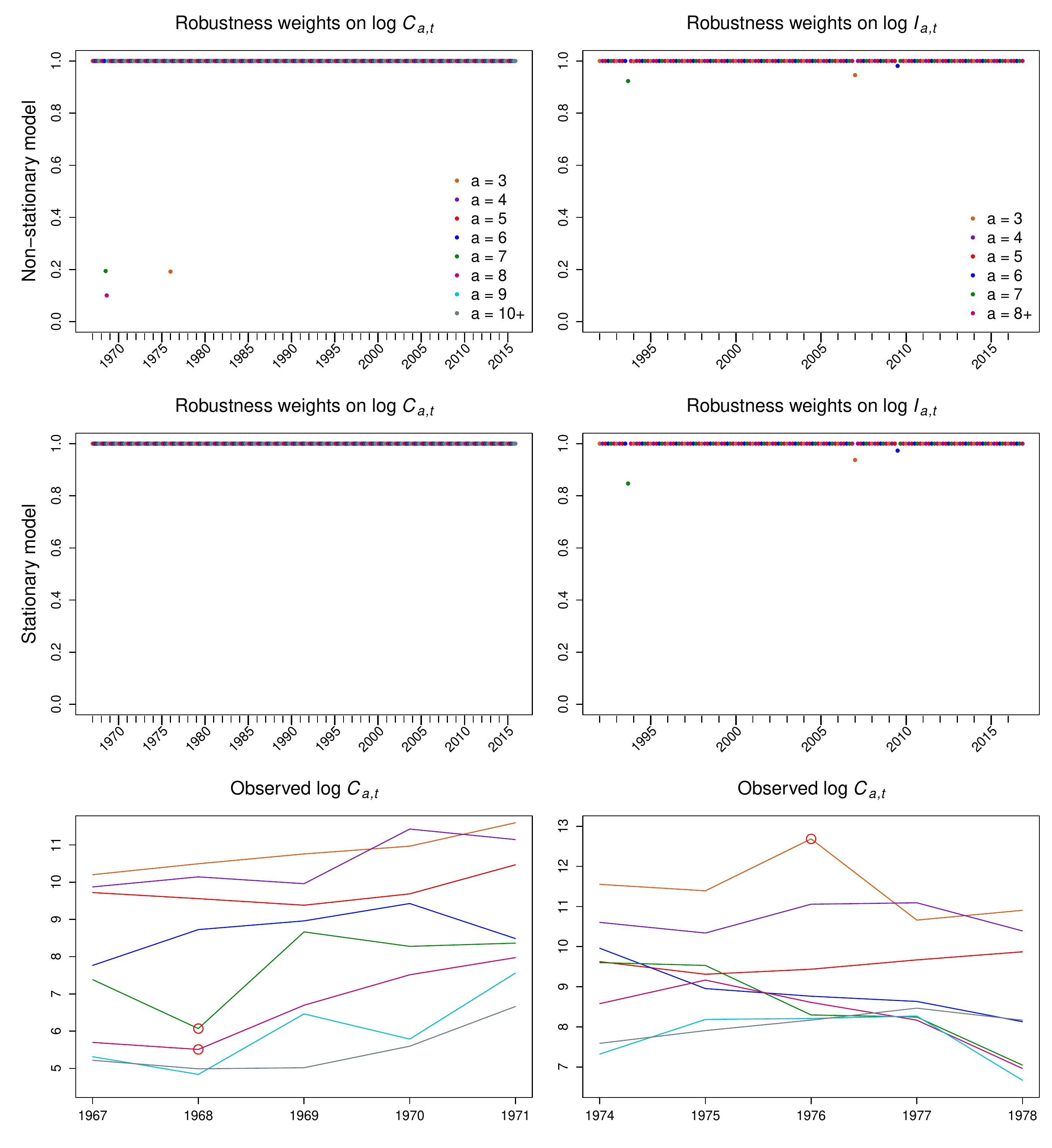}
\end{center}
\caption{North Sea pollock data robustness weights on the likelihood contributions of the commercial catches $\log C_{a,t}$ and survey indices $\log I_{a,t}$ for both non-stationary and stationary versions of the model (top and middle rows), and observed $\log C_{a,t}$ series where circled data points are the most downweighted under the non-stationary model (bottom row). \label{fig:dataweights}}
\end{figure}


\FloatBarrier
\newpage 




\bibliographystyle{plainnat} 
\bibliography{../../../GlobalRefs}

\end{document}